\documentclass[11pt]{article}

\usepackage{amsfonts, amsthm}
\usepackage{amssymb}
\usepackage{amsmath}  
\usepackage{graphicx, comment, xcolor}
\usepackage{authblk}
\usepackage{mathrsfs}
\usepackage[colorlinks=true, allcolors=blue]{hyperref}
\usepackage{cleveref}
\usepackage{mathtools}

\usepackage[T1]{fontenc}

\usepackage[shortlabels]{enumitem}
\usepackage{tikz}
\usetikzlibrary{positioning,calc}
\usetikzlibrary{intersections}
\tikzset{every picture/.style={line width=0.75pt}} %set default line width to 0.75pt     

\usepackage{fullpage}

\newtheorem{theorem}{Theorem}
\newtheorem{lemma}{Lemma}
\newtheorem{corollary}{Corollary}
\newtheorem{proposition}{Proposition}
\newtheorem{definition}{Definition}

\newcommand{\tr}{\mathrm{tr}}
\newcommand{\x}{\mathbf{x}}
\newcommand{\p}{\mathbf{p}}
\newcommand{\ac}{\mathbf{a}}

\newcommand{\R}{\mathbf{R}}

\newcommand{\dd}{\text{\rm{d}}}
\newcommand{\PiInner}[2]{\pi_{#1}^{#2}}
\newcommand{\diag}{\text{{\rm diag}}}

\newcommand{\ket}[1]{|#1\rangle}

\newcommand{\ketbra}[2]{|#1\rangle\langle#2|}
\newcommand{\bfa}{\mathbf{a}}

\title{Towards Optimal Convergence Rates for the Quantum Central Limit Theorem}

\author{Salman Beigi$^{1, 2}$, Hami Mehrabi$^3$}
\affil{\it \footnotesize $^1$School of Mathematics, Institute for Research in Fundamental Sciences (IPM), P.O. Box 19395-5746, Tehran, Iran \\
\it \footnotesize $^2$Centre for Quantum Technologies, National University of Singapore, Singapore\\
\it \footnotesize $^3$Department of Mathematics, Sharif University of Technology, Tehran, Iran}

\begin{document}
\maketitle

\begin{abstract}
The quantum central limit theorem for bosonic quantum systems states that the sequence of states $\rho^{\boxplus n}$ obtained from the $n$-fold convolution of a centered quantum state $\rho$ converges to a quantum Gaussian state $\rho_G$ that has the same first and second moments as $\rho$. In this paper, we contribute to the problem of finding the optimal rate of convergence for this quantum central limit theorem. We first show that if an $m$-mode quantum state has a finite moment of order $\max\{3, 2m\}$, then we have $\|\rho^{\boxplus n} - \rho_G\|_1=\mathcal O(n^{-1/2})$. We also introduce a notion of Poincar\'e inequality for quantum states and show that if $\rho$ satisfies this Poincar\'e inequality, then $D(\rho^{\boxplus n}\| \rho_G)= \mathcal O(n^{-1})$. By giving an explicit example, we verify that both these convergence rates are optimal.
\end{abstract}

{\footnotesize
\tableofcontents
}

\section{Introduction}

Central Limit Theorem (CLT) stands as one of the most notable theorems in probability theory. Berry-Esseen's  CLT~\cite{BhRa} states that, assuming finiteness of the third moment, the cumulative distribution function of the normalized sum $Y_n=\frac{X_1+\cdots + X_n}{\sqrt n}$ of i.i.d random variables $X_1, \dots, X_n$ with zero mean and non-zero second moment, uniformly converges to that of a Gaussian distribution $Z$ at the rate of $\mathcal O\big(\frac{1}{\sqrt n}\big)$. Berry-Esseen's theorem, however, does not have any implication on the convergence rate in total variation distance (1-norm). 

We note that convergence in total variation distance is much stronger than the point-wise convergence of cumulative distribution functions. In particular, when $X_k$'s are discrete, their normalized sum $Y_n=\frac{X_1+\cdots + X_n}{\sqrt n}$ remains discrete for any $n$, and its total variation distance to any continuous distribution, particularly to a Gaussian distribution, does not decay. Nevertheless, as long as the density of $X_k$'s contains an \emph{absolutely continuous part} and it has a finite third moment, we have a CLT in terms of total variation distance with the rate of convergence $\|Y_n-Z\|_{1}=\mathcal O\big(\frac{1}{\sqrt n}\big)$. Relaxing the assumption of finite third moment, more refined versions have also been established; see, e.g.,~\cite{MaSi, BhRa}.

Barron in 1986 proved an \emph{entropic CLT} in~\cite{Barron1986} where he showed that the relative entropy $D(Y_n \| Z)$ between $Y_n=\frac{X_1+\cdots + X_n}{\sqrt n}$ and the associated Gaussian random variable $Z$   converges to zero, assuming that the entropy of $X_k$'s are finite. Barron's result, however, does not say anything about the rate of convergence. This remained a challenge until Artstein, Ball, Barthe and Naor showed in 2004 that $D(Y_n \| Z)$ is indeed decreasing~\cite{ABBN1}. This was later reproved by several authors; see~\cite{Courtade} and references therein. Then the same group showed $D(Y_n \| Z)=\mathcal O\big(\frac{1}{ n}\big)$, assuming that the starting distribution satisfies a \emph{Poincar\'e inequality}~\cite{ABBN2}. Johnson and Barron~\cite{JB} also proved this bound under the same assumption, but with a completely different method. We note that the assumption that the random variables satisfy the Poincar\'e inequality is a strong constraint since it implies the finiteness of all the moments. Thus, it remained a challenge to relax this assumption. Finally, Bobkov, Chistyakov and G{\"o}tze proved in 2014 that $D(Y_n \| Z)=\mathcal O\big(\frac{1}{ n}\big)$ holds true assuming the finiteness of the fourth moment and the finiteness of entropy~\cite{BCG}.

\subsection{Quantum CLT}
The convolution of probability distributions can be generalized for quantum states over bosonic modes. Having two quantum states $\rho, \sigma$ over $m$ bosonic modes, their convolution is another $m$-mode state denoted by $\rho\boxplus \sigma$ that is defined in terms of their interaction via a beam splitter; see Subsection~\ref{subsec:convolution} for the precise definition. Then, we may consider $\rho^{\boxplus n}$, the $n$-fold convolution of $\rho$ with itself.  Now, it is natural to ask if there is a quantum version of CLT, i.e., whether $\rho^{\boxplus n}$ converges to a Gaussian state. 

Cushen and Hudson~\cite{CH} in 1971 proved that if $\rho$ is centered (i.e., it has zero first moments) and has finite second moments, then $\rho^{\boxplus n}\to \rho_G$ \emph{weakly}, where $\rho_G$ is the Gaussian state with the same first and second moments as $\rho$. Here, weak convergence means that for any bounded operator $X$,
		$$\tr\big(\rho^{\boxplus n} X\big) \rightarrow \tr\big(\rho_G X\big).$$

Becker, Datta, Lami and Rouz\'e~\cite{CRLimitTheorem} have recently explored the problem of convergence rate in the above quantum CLT. They showed that  if $\rho$ has a finite third-order moment, then
\[
	\big\| \rho^{\boxplus n} - \rho_G \big\|_{2} = \mathcal{O} \Big(\frac{1}{\sqrt n}\Big),
\]
where ${\| . \|}_2$ denotes the Hilbert-Schmidt norm. Moreover, if $\rho$ has a finite fourth-order moment and the third derivative of its characteristic function at origin vanishes, then
\[
	\big\| \rho^{\boxplus n} - \rho_G \big\|_{2} = \mathcal{O} \Big(\frac{1}{n}\Big).
\]
By numerical simulations, it is also shown in~\cite{CRLimitTheorem} that these convergence rates are tight. However, it is more desirable to find tight convergence rates when the figure of merit is $1$-norm or relative entropy as in the aforementioned results in the classical case since they have operational meanings. Interestingly, based on the above bounds on the Hilbert-Schmidt norm, and using the so called gentle measurement lemma and continuity bounds for the entropy function, it is also proven in~\cite{CRLimitTheorem} that
\[\big\| \rho^{\boxplus n} - \rho_G \big\|_{1} = \mathcal{O} \Big(n^{-\frac{1}{2(m+1)}}\Big),  \]
and
\[
	D(\rho^{\boxplus n} \| \rho_G ) = \mathcal{O} \Big(n^{-\frac{1}{2(m+1)}}\log n\Big),
\]
where $m$ is the number of modes.  These bounds, although interesting, seem suboptimal at least in comparison to known results in the classical case, even if $m=1$.  

Following the work of Cushen and Hudson, various other forms of CLT have also been established in the non-commutative setting. We refer to~\cite{CRLimitTheorem} for a list of these works. We also refer to~\cite{BuGuJaffe2023} for a more recent work on the subject.

\subsection{Main results}
In this paper we prove two main theorems regarding the convergence rates in the quantum CLT. We first show that if $\rho$ is centered and has a finite moment of order $\kappa=\max\{2m, 3\}$, then
\begin{equation} \label{Result1}
	{\| \rho^{\boxplus n} -  \rho_G \|}_1 = \mathcal{O} \Big(\frac{1}{\sqrt n}\Big).
\end{equation}
This bound matches the optimal convergence rate in the classical case, yet with a stronger assumption. Note however that for $m=1$, we only assume finiteness of the third moment. By giving an explicit example, we show that the above convergence rate cannot be improved even if all the moments of $\rho$ are finite.  

Our second main result is an entropic quantum CLT. Following the classical works~\cite{ABBN2, JB}, we assume that the quantum state $\rho$ satisfies a \emph{quantum Poincar\'e inequality} and prove that
\begin{equation} \label{Result2}
	D(\rho^{\boxplus n} \| \rho_G ) = \mathcal{O} \Big(\frac{1}{n}\Big).
\end{equation}
This bound matches the convergence rate of the entropic CLTs of~\cite{ABBN2, JB} in  the classical case.

To prove~\eqref{Result1}, we employ quite standard ideas in the classical theory of CLTs and generalize them to the quantum case. Our first step of the proof is to apply H\"older's inequality to bound $1$-norm in terms of $2$-norm. Then, similar to~\cite{CRLimitTheorem} we use the \emph{quantum Plancherel identity} to express the $2$-norm in terms of the $2$-norm of \emph{characteristic functions} and their derivatives. Since characteristic functions are well-behaved under convolution as in the classical case, it is not hard to work with their $2$-norm. Then, a careful analysis of the $2$-norm of the characteristic functions gives~\eqref{Result1}.

The proof of our second result is more involved. We follow similar steps as in the work of Johnson and Barron~\cite{JB} to prove~\eqref{Result2}. The first step in~\cite{JB} is to use \emph{de Bruijn's identity} to express $D(Y_n \| Z)$ in terms of \emph{Fisher information distance} $I(Y_n)-I(Z)$, reducing the proof of $D(Y_n \| Z)=\mathcal O\big(\frac{1}{ n}\big)$ to that of $I(Y_n)-I(Z)=\mathcal O\big(\frac{1}{ n}\big)$. This is already a challenge in the quantum case since there is a plethora of quantum generalizations of Fisher information~\cite{PETZ1996, LesniewskiRuskai99}. Among such generalizations, a quantum de Bruijn's identity has been proven for the so called \emph{Kubo-Mori-Bogoliubov} Fisher information~\cite{KS}. However, this quantum Fisher information is hard to work with and does not satisfy some other desirable properties. To this end, we use another quantum generalization of Fisher information called  the \emph{Symmetric Logarithmic Derivative (SLD)} Fisher information; see, e.g.,~\cite{Hayashi12002Fisher} and references therein.  We use the SLD Fisher information denoted by $I(\cdot)$ to generalize the method of~\cite{JB}. Assuming some quantum Poincar\'e inequality, we prove that $I(\rho^{\boxplus n}) - I(\rho_G) = \mathcal O\big(\frac 1n\big)$. Next, we diverge from the method of~\cite{JB} and use a \emph{quantum logarithmic-Sobolev inequality} for the quantum Ornstein-Uhlenbeck semigroup from~\cite{Beigi-RahimiKeshari2023} to bound $D(\rho^{\boxplus n} \| \rho_G )$ with the SLD Fisher information distance. Putting these together our desired bound~\eqref{Result2} is implied.

\paragraph{Layout of the paper.} In Section \ref{secND} we introduce some of the main notions and tools required for our work. In particular, we define quantum convolution and introduce the notion of quantum Poincar\'e inequality. Section~\ref{SecTraceDistanceProof} includes our first main result stated in Theorem~\ref{TraceMainTheoremM}, that is a quantum CLT in terms of $1$-norm. Section~\ref{SimplifiedFisherInformation} is devoted to our entropic CLT and our second result stated in Theorem~\ref{MainTheorem}. In this section, we define the SLD Fisher information, and show that using logarithmic-Sobolev inequalities, the problem of entropic CLT can be reduced to the problem of the convergence of the SLD Fisher information distance. The proof of the convergence rate for the SLD Fisher information distance is given in the following sections. First, in Section~\ref{SuccessiveProjection} we provide a quantum generalization of a main tool in~\cite{JB} called the \emph{successive projection lemma.} Then, in Section~\ref{ConvergenceFID} we use this lemma to prove our next result stated in Theorem~\ref{pseudoMainTheorem-v2} about the convergence rate for the SLD Fisher information distance. Final remarks come in Section~\ref{secCo} and some details are left for the appendices.

\section{Preliminaries} \label{secND}

In this section, we review some basic notions of continuous quantum systems. We refer to~\cite{Serafini} for more details on the subject.

The Hilbert space of $m$ bosonic modes (harmonic oscillators) is the space of complex square-integrable functions on $\mathbb{R}^m$, hereafter denoted by $\mathcal{H}_{m} := L^{2}(\mathbb{R}^m)$. 
Let $\x_j, \p_j$ be the canonical quadrature (position and momentum)  operators acting on the $j$-th mode. Then, the annihilation and creation operators are given by
\[\ac_j  := \frac{1}{\sqrt{2}}(\x_j+ i \p_j), \qquad  \ac_j^{\dagger} := \frac{1}{\sqrt{2}}(\x_j - i \p_j).\]
Fock states $|n_1\rangle\otimes \cdots \otimes |n_m\rangle=|{n_1, \dots, n_m}\rangle$, $n_1, \dots, n_m\in \{0, 1, 2, \dots\}$ form an orthonormal basis for $\mathcal H_m$, and we have $\ac_j|n_j\rangle= \sqrt {n_j}|n_j-1\rangle$ and $\ac_j^\dagger |n_j\rangle = \sqrt{n_j+1}|n_j+1\rangle$.
What is important about these operators is that they satisfy the \emph{canonical commutation relations:}
\begin{equation} \label{CCR}
	\big[\ac_j, \ac_k^{\dagger}\big] = \delta_{jk}, \qquad\quad [\ac_j, \ac_k] = 0,
\end{equation}
where $[X, Y] = XY-YX$ and $\delta_{jk}$ is the Kronecker delta function.\footnote{Here, and in the following, by $\delta_{jk}$ we mean $\delta_{jk}\mathbb I$ where $\mathbb I$ denotes the identity operator.} 
We denote the column vector of position and momentum operators by $\R$:
\[
	\R = (\x_1, \p_1, ... , \x_m, \p_m)^{\top}. 
\]
Then, canonical commutation relations~\eqref{CCR} can be summarized as 
\[	[\R, \R^T] = i \Omega_m, \]
where $[\R, \R^T]$ is understood as the matrix of coordinate-wise commutations, and $\Omega_m$ is defined by
\begin{align}\label{eq:def-Omega}
\Omega_m = \bigoplus_{j=1}^{m} \Omega_1, \qquad \quad	\Omega_1 = \begin{pmatrix}
	 0 & 1 \\
	  -1 & 0 
	  \end{pmatrix}.
\end{align}

An $m$-mode quantum state $\rho$ is an operator acting on $\mathcal H_m$ that is a density operator meaning that it is positive semidefinite and satisfies $\tr(\rho)=1$.  Rank-one density operators are called \emph{pure states}. $\rho$ is called faithful if its kernel is trivial.

In this paper, we often work with \emph{unbounded operators} acting on $\mathcal H_m$ including creation and annihilation operators. Such an unbounded operator should be specified by its domain that is a \emph{dense} subspace of $\mathcal H_m$. We refer to~\cite{Hall-Book} for more details on unbounded operators, and particularly for the domain of the canonical operators $\bfa_j, \bfa_j^\dagger$.

Any density operator $\rho$ has a spectral decomposition $\rho=\sum_\ell c_\ell \ketbra{v_\ell}{v_\ell}$. Then for a possibly unbounded operator $X$ we let $\tr(\rho X^\dagger X)=\tr(X^\dagger X\rho)=\tr(X\rho X^\dagger )$ where 
$$\tr(X\rho X^\dagger ) = \sum_\ell c_\ell \|X\ket{v_\ell}\|^2,$$
if $\ket{v_\ell}$ belongs to the domain of $X$ for any $\ell$, and $\tr(\rho X^\dagger X)=+\infty$ otherwise. In fact, $\tr(\rho X^\dagger X)=  \tr(X^\dagger X\rho)$ is set to be equal to the trace of $X\rho X^\dagger$ if it is trace class, and equal to $+\infty$ otherwise. 
Generalizing this definition, for a Hermitian operator $A$, there are positive matrices $A_+, A_-$ such that $A=A_+-A_-$ and then we can define $\tr(\rho A) = \tr(\rho A_+) - \tr(\rho A_-)$ if $\tr(\rho A_{\pm})=\tr(\sqrt{A_{\pm}} \rho \sqrt{A_{\pm}}^\dagger )<+\infty$. We emphasize that here $\tr(\rho A)<+\infty$ does not necessarily mean that $\rho A$ is a trace class operator, yet as is a convention in quantum physics $\tr(\rho A)$ is understood as the expectation of the quantum observable $A$ when the underlying quantum state is $\rho$~\cite{busch2016quantum}. We note that by our convention, at least when $\tr(\rho A)= \tr(A\rho)<+\infty$, we have $\tr([\rho, A])=0$.  All quantum states considered in this paper are assumed to satisfy $\tr(\rho \bfa_j^\dagger \bfa_j)<+\infty$ for any $j$. Equivalently, we assume that $\bfa_j \rho \bfa_j^\dagger$ is trace class for any $j$.

\subsection{Displacement operators and characteristic function}

For $z = (z_1, z_2, \dots , z_{m})^{\top} \in \mathbb{C}^{m}$, the $m$-mode displacement (Weyl) operator $D_z$ is defined by
\[
	D_{z} =\exp \bigg( \sum_{j=1}^{m} (z_j \ac_{j}^{\dagger} - \bar{z_j} \ac_{j})\bigg).
\]
The canonical commutation relations~\eqref{CCR} and the Baker--Campbell--Hausdorff (BCH) formula yield
\[
	D_z D_w = e^{\frac{1}{2} (z^{\top} \bar w - \bar z^{\top} w)} D_{z+w}.
\]
Therefore, $D_z D_z^{\dagger} = D_z D_{-z} = D_0 = \mathbb{I}$ and $D_z$ is unitary. Displacement operators displace quadrature operators under conjugation:
\begin{equation} \label{DisplacementOp}
	D_z \ac_j D_z^{\dagger} = \ac_{j} - z_j.
\end{equation}

The \emph{characteristic function} $\chi_{T}: \mathbb{C}^{m} \rightarrow \mathbb{C}$ of a \emph{trace class} operator $T$ is defined by
\[
	\chi_{T}(z)  \coloneqq \tr\big(T D_z\big).
\] 
The significance of the characteristic function is that it fully determines the trace class operator as
\[
	T = \frac{1}{{\pi}^m} \int_{\mathbb{C}^{m}} \chi_{T}(z) D_{-z} \dd^{2m} z.
\]
Since displacement operators are unitary and bounded, the characteristic function of a trace class operator is bounded. Moreover, it is uniformly continuous~\cite{Holevo, CRLimitTheorem}. 

The \emph{Wigner function} of a trace class operator $T$ is defined as the Fourier transform of its characteristic function:
\[
W_{T}(z) := \frac{1}{\pi^{2m}} \int \chi_{T}(w) e^{z^{\top} \bar w - \bar z^{\top} w} \dd^{2m} w.	
\] 
It is not difficult to verify that $W_{T^{\dagger}}(z) = \overline{W_{T}(z)}$. Thus, for a quantum state $\rho$, we have $W_{\rho}(z) \in \mathbb{R}$. 

The characteristic function of a quantum state has the crucial property that it attains value $1$ in absolute value only at $z=0$.  Even more, it is shown in~\cite[Proposition 14]{CRLimitTheorem} that for any $\epsilon > 0$, there exists $\delta > 0$ such that
\begin{equation} \label{SupChar}
	\sup_{z \in \mathbb{C}^m \setminus B(0, \epsilon)} |\chi_{\rho} (z)| = 1 - \delta <1,
\end{equation}
where $B(0, \epsilon)\subset \mathbb C^{m}$ is the ball of radius $\epsilon$ around the origin.

\subsection{Moments of quantum states}
The \emph{vector of first-order moments} of an $m$-mode quantum state $\rho$ is defined by
\[
	\mathbf{d}(\rho) = \tr\big(\rho \R\big) = \big(\tr(\rho\x_1), \tr(\rho\p_1), \dots ,  \tr(\rho\x_m) , \tr(\rho \p_m)\big)^{\top} \in \mathbb R^{2m}.
\]
$\rho$ is said to be \emph{centered} if $\mathbf{d}(\rho) =0$ or equivalently 
\[
	\tr[\rho \ac_{j}] = 0, \qquad \forall j \in \{ 1, \dots , m\}.
\]
The \emph{covariance matrix} of $\rho$ is a $2m \times 2m$ matrix given by
\[
	\boldsymbol{\gamma}(\rho) = \tr\Big(\rho \big\{ \R - \mathbf{d}(\rho), {(\R - \mathbf{d}(\rho))}^{\top} \big\}\Big) = \tr\Big( \rho \{ \R , \R^\top \}\Big) - 2 \mathbf{d}(\rho) \mathbf{d}(\rho)^{\top},
\]
where $\{X, Y\} = XY+YX$ stands for anti-commutation. By definition, $\boldsymbol{\gamma}(\rho)$ is a real and symmetric matrix. 
Moreover, as a consequence of~\eqref{CCR} we have
\begin{equation} \label{eq:Uncertainty}
	\boldsymbol{\gamma}(\rho) + i \Omega_m \geq 0,
\end{equation}
where $\Omega_m$ is given in~\eqref{eq:def-Omega}. In particular, $\boldsymbol{\gamma}(\rho)$ is positive definite; see~\cite{Serafini} for more details.

In this paper we consider only quantum states that have finite covariance matrices, i.e., the entries of the covariance matrices are finite. Yet,
we are interested in higher moments of a quantum state. To this end, let the \emph{photon number operator} $H_m$ be 
\[
		H_{m} := \sum_{j=1}^{m} \ac_{j}^{\dagger} \ac_{j}.
\]
Then, following~\cite{CRLimitTheorem} for a density operator $\rho$ we define its \emph{moment of order $\kappa>0$} by
\[
		M_{\kappa}(\rho) := \tr \Big(\rho (H_{m}+m)^{\frac 12 \kappa }\Big).
\]
Moments of a quantum state are related to the derivatives of its characteristic function~\cite[Proposition 25]{CRLimitTheorem}. Indeed, for any $\epsilon > 0$, there is a constant $c_{\kappa,m}(\epsilon) < \infty$ such that for every $m$-mode quantum state $\rho$ we have
\begin{align}\label{eq:moments-derivatives}
	\max_{|\alpha| + |\beta| \leq \kappa} \sup_{z \in B(0, \epsilon)} \big|\partial_{z}^{\alpha} \partial_{\bar z}^{\beta} \chi_\rho(z)\big| \leq c_{\kappa,m}(\epsilon) M_{\kappa}(\rho).
\end{align}
This means that if $\rho$ has a finite $\kappa$-th order moment, then all the derivatives of its characteristic function of order up to $\kappa$ are bounded around the origin.

\subsection{Gaussian states}\label{subsec:Gaussian-states}
A density operator acting on $\mathcal{H}_m$ is called a \emph{Gaussian state} if its Wigner function is a Gaussian probability distribution on the real space $\mathbb{R}^{2m}$. Equivalently, a state is Gaussian if its characteristic function is a Gaussian function (exponential of a quadratic polynomial). A Gaussian state is fully determined by its vector of first-order moments and its covariance matrix. In fact, for any $d\in \mathbb R^{2m}$ and any $2m \times 2m$ symmetric real matrix $\boldsymbol{\gamma}$ satisfying \eqref{eq:Uncertainty}, there is a Gaussian state with the vector of first-order moments $d$ and covariance matrix $\boldsymbol{\gamma}$.

\begin{definition}\label{def:Gaussification}
The \emph{Gaussification} of a quantum state $\rho$ with finite second moments, denoted by $\rho_G$, is a Gaussian state whose vector of first-order moments and covariance matrix are $\mathbf{d}(\rho)$ and $\boldsymbol{\gamma}(\rho)$, respectively.
\end{definition}

A unitary operator acting on $\mathcal H_m$ is called Gaussian if it sends Gaussian states to Gaussian states under conjugation. Since a Gaussian states $\rho$ is determined by $\mathbf{d}(\rho)$ and $\boldsymbol{\gamma}(\rho)$, such a Gaussian unitary is characterized by its action on the vector of first-order moments and the covariance matrix. Indeed, for any $r\in \mathbb R^{2m}$ and any $2m \times 2m $ \emph{symplectic matrix} $S$, i.e., a matrix that satisfies
\[
	S \Omega_m S^{\top} = \Omega_m,
\]
there is a corresponding unitary $U_{S, r}$ that is Gaussian and 
	\[
		\mathbf{d}\big( U_{S,r}^{\dagger} \rho U_{S,r}\big) = S \mathbf{d}(\rho) - r,
	\]
	\[
		\boldsymbol{\gamma}\big(U_{S,r}^{\dagger} \rho U_{S,r}\big) = S \boldsymbol{\gamma}(\rho) S^{\top}.
	\]
Here, the condition that $S$ has to be symplectic is imposed due  to the canonical commutation relations. Such a Gaussian unitary can be decomposed into a Gaussian unitary $U_{S}$ that depends only on $S$ and a displacement operator. Thus, the classification of Gaussian unitaries is reduced to that of unitaries of the form $U_{S}$~\cite{Serafini}.

Williamson's theorem~\cite{William} states that for any $2m \times 2m$ real symmetric matrix $\boldsymbol{\gamma}$, there exists a symplectic matrix $S$ such that,
\[
	S \boldsymbol{\gamma} S^{\top} = \diag(\nu_1, \nu_1, \dots , \nu_m, \nu_m).
\]
The following proposition is an immediate consequence of Williamson's theorem and the above discussion.

\begin{proposition} \label{Will}
For any $m$-mode quantum state $\rho$ there exists a Gaussian unitary $U$ such that
	\begin{align}
		& \mathbf{d}(U \rho U^{\dagger}) = 0,\\
		& \boldsymbol{\gamma}(U \rho U^{\dagger}) = \diag(\nu_1, \nu_1, \dots , \nu_m, \nu_m), \label{eq:Williamson-diag}
	\end{align}
	where, due to~\eqref{eq:Uncertainty}, $\nu_j\geq 1$. We call $U \rho U^{\dagger}$ a Williamson's form of $\rho$.
\end{proposition}

If $\rho$ is in Williamson's form, then $\rho_G=\tau_1\otimes \cdots\otimes \tau_m$ is a tensor product of \emph{thermal states} of the form
\begin{align}\label{eq:thermal}
\tau_j = \big(1-e^{-\beta_j}\big)  e^{-\beta_j \ac_j^\dagger \ac_j},
\end{align}
where $\beta_j>0$ is defined by
\begin{align}\label{eq:nu-beta}
\nu_j= \frac{1+ e^{-\beta_j}}{1- e^{-\beta_j}}\geq 1.
\end{align}

Note that if $\nu_j=1$, then $\beta_j=+\infty$ and $\tau_j$ is the \emph{vacuum state} $\ket 0 \langle 0|$ which is pure. We also note that in this case even in $\rho$ the $j$-th mode is in the vacuum state. The point is that $\nu_j=1$ implies $\tr(\rho_G \bfa_j^\dagger \bfa_j) = 0$ and since the covariance matrix of $\rho$ matches that of $\rho_G$, we also have $\tr(\rho \bfa_j^\dagger \bfa_j) =0$. Thus, the support of $\rho$ is included in the kernel of $\bfa_j^\dagger \bfa_j$. This means that the $j$-th mode in $\rho$ is in the vacuum state, which is pure so the $j$-mode is decoupled from the other modes.

The characteristic and Wigner functions of $\rho_G=\tau_1\otimes \cdots\otimes \tau_m$ are 
\begin{align}\label{eq:tau-chi-W}
\chi_{\rho_G}(z) = \prod_{j=1}^m e^{-\frac 12 \nu_j|z_j|^2}, \qquad \quad W_{\rho_G}(z)= \prod_{j=1}^m \frac{2}{\pi \nu_j} e^{-\frac{2}{\nu_j}|z_j|^2}.
\end{align}

\subsection{Relative entropy and $p$-norm}\label{subsec:p-norm}
Given two quantum states $\rho, \sigma$ where the support of $\rho$ is included in the support of $\sigma$, the quantum relative entropy between $\rho, \sigma$ is defined~\cite{Um1962} by
\[
	D(\rho \| \sigma ) := \tr\big(\rho(\log \rho - \log \sigma)\big).
\]

For any operator $T$, let $|T| = \sqrt{T^\dagger T}$. Then, for $p\geq 1$, assuming that $|T|^p$ is a trace class operator, the Schatten $p$-norm of $T$ is defined  by
\[
	\| T\|_{p} : =\Big(\tr\big(|T|^{p}\big)\Big)^{1/p}.
\]
When $p=1$, this norm is sometimes called the \emph{trace norm}.  Schatten norms satisfy H\"older's inequality.

The $p$-norm for $p=2$ it is called the \emph{Hilbert-Schmidt norm}. Interestingly, the Hilbert-Schmidt norm can be expressed in terms of the characteristic function:
\begin{equation} \label{PlancherelEq}
	\| T\|^{2}_2 = \frac{1}{\pi^m} \int \big| \chi_{T}(z)\big|^{2} \dd^{2m}z.
\end{equation}
We refer to the above equation as the \emph{quantum Plancherel identity}.

The $p$-norm in the limit of $p\to+\infty$ recovers the operator norm and is denoted by $\|\cdot\|=\|\cdot\|_\infty$.

\subsection{Quantum convolution}\label{subsec:convolution}
The beam splitter is a particular Gaussian unitary that acts on a bipartite quantum system with each part consisting of $m$ modes. The beam splitter with transmissivity parameter $\eta \in [0,1]$ is defined by
\[
	U_{\eta} := \exp\bigg(\arccos(\sqrt{\eta}) \sum_{j=1}^{m}\big(\ac_{j,1}^{\dagger} \ac_{j,2} - \ac_{j,1} \ac_{j,2}^{\dagger}\big)\bigg) = \prod_{j=1}^{m} \exp\Big(\arccos(\sqrt{\eta}) \big(\ac_{j,1}^{\dagger} \ac_{j,2} - \ac_{j,1} \ac_{j,2}^{\dagger}\big)\Big),
\]
where $\ac_{j,1}$ is the annihilation operator of the $j$-th mode of the first subsystem, and  $\ac_{j,2}$ is the annihilation operator of the $j$-th mode of the second subsystem. 
$U_\eta$ can also be characterized in terms of its action on these annihilation operators:
\begin{align} 
\begin{cases} \label{BeamA}
	U_{\eta} \ac_{j,1} U_{\eta}^{\dagger} = \sqrt{\eta} \ac_{j,1} - \sqrt{1 - \eta} \ac_{j,2},\\
	U_{\eta} \ac_{j,2} U_{\eta}^{\dagger} =  \sqrt{1 - \eta} \ac_{j,1} + \sqrt{\eta} \ac_{j,2} . 
\end{cases}
\end{align}

The \emph{quantum convolution} of two $m$-mode quantum states $\rho, \sigma$ with parameter $\eta\in [0,1]$ is defined by 
\[
	\rho \boxplus_{\eta} \sigma = \tr_2\Big(U_{\eta} (\rho \otimes \sigma ) U_{\eta}^{\dagger}\Big),
\]
where $\tr_{2}(\cdot)$ stands for the partial trace with respect to the second subsystem. Using~\eqref{BeamA} it can be verified that
\begin{align}\label{eq:convolution-char-function}
	\chi_{\rho \boxplus_{\eta} \sigma}(z) = \chi_{\rho}(\sqrt{\eta} z)\, \chi_{\sigma}(\sqrt{1-\eta} z) .
\end{align}
A crucial property of quantum convolution is that it commutes with Gaussian unitaries $U$ as follows:
\begin{equation} \label{UnitaryAndConv}
	U (\rho \boxplus_{\eta} \sigma) U^{\dagger} = (U \rho U^{\dagger}) \boxplus_{\eta} (U \sigma U^{\dagger}).
\end{equation}
This equation can be proven using~\eqref{BeamA}.

In this paper, we are interested in the \emph{symmetric} quantum convolution of $n$ quantum states. Let $\rho_1, \dots, \rho_n$ be $m$-mode quantum states. Then, their symmetric quantum convolution denoted by $\rho_1 \boxplus \dots \boxplus \rho_n$ is defined inductively by $\rho_1 \boxplus \rho_2 = \rho_1 \boxplus_{\frac 12} \rho_2$ and 
\begin{align*}
	 \rho_1 \boxplus \dots \boxplus \rho_n = (\rho_1 \boxplus \cdots \boxplus \rho_{n-1}) \boxplus_{1 - \frac 1n} \rho_n.
\end{align*}
When $\rho_1=\cdots =\rho_n=\rho$, we denote $\rho_1 \boxplus \dots \boxplus \rho_n$ by $\rho^{\boxplus n}$. Using~\eqref{eq:convolution-char-function} it is readily verified that 
\[
	\chi_{\rho_1 \boxplus \cdots \boxplus \rho_n}(z) = \chi_{\rho_1}\Big(\frac {z}{\sqrt n}\Big) \cdots \chi_{\rho_n}\Big(\frac {z}{\sqrt n}\Big),
\]
and
\[
	\chi_{\rho^{\boxplus n}}(z) = \chi_{\rho}\Big(\frac {z}{\sqrt n}\Big)^{n}.
\]
Also, using~\eqref{BeamA} it can be shown that 
\begin{align}\label{eq:gamma-d-conv}
\boldsymbol{\gamma}(\rho\boxplus_\eta \sigma) = \eta \boldsymbol{\gamma}(\rho) + (1-\eta)\boldsymbol{\gamma}(\sigma), \qquad \mathbf{d}(\rho\boxplus_\eta\sigma) = \sqrt \eta  \mathbf{d}(\rho) + \sqrt{1-\eta} \mathbf{d}(\sigma).  
\end{align}
The above equations imply that if $\rho$ is a Gaussian state with $\mathbf{d}(\rho)=0$, then $\rho^{\boxplus n}= \rho$ for all~$n$.

The following lemma states a crucial property of quantum convolution that is heavily used in this paper.

\begin{lemma} 
Let $\rho$ and $\sigma$ be two $m$-mode quantum states with finite second moments and let $\eta \in [0,1]$ be a transmissivity parameter. Let $\ac_j$ be the annihilation operator of the $j$-th mode of the output of the convolution map. Then, we have
\[
	\sqrt{\eta} \big[\ac_{j}, \rho \boxplus_{\eta} \sigma\big] = \big[\ac_{j,1},\, \rho\big] \boxplus_{\eta} \sigma.
\]
\end{lemma}

Here, $\big[\ac_{j,1}, \rho\big] \boxplus_{\eta} \sigma$ is defined simply by extending the definition of convolution to arbitrary trace class operators. We note that by H\"older's inequality 
$$\|\rho \ac_{j,1}\|_1 \leq \| \sqrt \rho\|_2\cdot \|\sqrt \rho \ac_{j, 1}\|_2 = \tr\big(\ac_{j, 1}^\dagger \rho \ac_{j, 1}\big)^{1/2} <+\infty,$$ 
where the last inequality follows from the fact that $\rho$ has finite second moments. Similarly, we have $\|\ac_{j, 1}\rho\|_1<+\infty$, so $[\ac_{j, 1}, \rho]$ is trace class and $\big[\ac_{j,1},\, \rho\big] \boxplus_{\eta} \sigma$ is well-defined.

\begin{proof}
Using~\eqref{BeamA} we compute
\begin{align*}
	[\ac_{j,1}, \rho] \boxplus_{\eta} \sigma & = \tr_2 \Big( U_{\eta} \big([\ac_{j,1} ,  \rho] \otimes \sigma \big) U_{\eta}^{\dagger} \Big)\\
	& = \tr_2 \Big( U_{\eta} \big([\ac_{j,1} ,  \rho\otimes \sigma]  \big) U_{\eta}^{\dagger} \Big)\\
	& = \tr_2 \Big(\big[ U_{\eta}\ac_{j,1}U_{\eta}^\dagger ,  U_{\eta} (\rho\otimes \sigma) U_{\eta}^\dagger  \big]  \Big)\\
	& = \tr_2 \Big(\big[ \sqrt{\eta} \ac_{j,1} - \sqrt{1 - \eta} \ac_{j,2} ,  U_{\eta} (\rho\otimes \sigma) U_{\eta}^\dagger  \big]  \Big)\\
& = \sqrt{\eta} \tr_2 \Big(\big[ \ac_{j,1} ,  U_{\eta} (\rho\otimes \sigma) U_{\eta}^\dagger  \big]  \Big) - \sqrt{1 - \eta} \tr_2 \Big(\big[  \ac_{j,2} ,  U_{\eta} (\rho\otimes \sigma) U_{\eta}^\dagger  \big]  \Big)\\
& = \sqrt{\eta} \tr_2 \Big(\big[ \ac_{j,1} ,  U_{\eta} (\rho\otimes \sigma) U_{\eta}^\dagger  \big]  \Big)\\
& = \sqrt{\eta}  \Big[ \ac_{j} ,  \tr_2 \big(U_{\eta} (\rho\otimes \sigma) U_{\eta}^\dagger  \big)  \Big]\\
& = \sqrt{\eta} \big[ \ac_{j} ,   \rho \boxplus_\eta \sigma \big].  
\end{align*}	
Here, in the sixth line we use the fact that $\ac_{j,2}$ acts on the second subsystem and both $\ac_{j,2}  U_{\eta} (\rho\otimes \sigma) U_{\eta}^\dagger, U_{\eta} (\rho\otimes \sigma) U_{\eta}^\dagger \ac_{j,2}$ are trace class operators to conclude that  $\tr_2 \Big(\big[  \ac_{j,2} ,  U_{\eta} (\rho\otimes \sigma) U_{\eta}^\dagger  \big] =0$. Also, in the penultimate equation we use the fact that $\ac_{j,1}$ acts on the first subsystem.

\end{proof}

By the definition of $\rho_1 \boxplus \cdots \boxplus \rho_n$, this lemma also gives
\begin{equation} \label{DeriveConvolution}
	\frac{1}{\sqrt{n}} \big[\ac_j, \rho_1 \boxplus \cdots \boxplus \rho_n \big] = [\ac_{j,1},\rho_1] \boxplus \rho_2  \cdots \boxplus  \rho_n.
\end{equation}

\subsection{Quantum Poincar\'e inequalities}\label{subsec:Poincare}
In this paper we need a generalization of the notion of Poincar\'e inequality for quantum bosonic systems. To this end, we first need to define an inner product on the space of operators with respect to a given quantum state $\rho$. 

Given a probability distribution on a measurable space, we may define the inner product of two functions $f, g$ on this space by $\mathbb E[\bar f g]$. Due to the non-commutative nature of quantum mechanics, this inner product can be generalized to the quantum setting in various ways. One natural choice that is used, e.g., in the context of quantum logarithmic-Sobolev inequalities is $(X, Y)\mapsto \tr(\rho^{1/2} X^\dagger \rho^{1/2}Y)$~\cite{OZ99} and is called the \emph{KMS inner product}. Another candidate, which as will be discussed later is the most natural one for our application, is the \emph{Kubo-Mori-Bogoliubov} inner product:
\begin{align}\label{eq:BKM-inner-prod}
(X, Y) \mapsto \int_0^1 \tr\Big( \rho^{1-s} X^\dagger \rho^s Y  \Big)\dd s.
\end{align}
Nevertheless, both the above inner products, particularly the latter one, are complicated to work with and do not satisfy some desirable properties. In fact, as will become clear later on, it is advantageous to pick an inner product that is \emph{linear} in $\rho$. To this end, we can take $(X, Y)\mapsto \tr(\rho X^\dagger Y)$. However, in some sense, this inner product is not stable enough to handle the case where the density operator $\rho$ is not faithful. To overcome this difficulty, we simply consider a symmetric version of this inner product defined by
\begin{align}\label{eq:symmetric-inner-prod}
\langle X, Y\rangle_\rho : = \frac{1}{2}\tr\big(\rho X^\dagger Y\big) + \frac{1}{2}\tr\big(X^\dagger \rho  Y  \big).
\end{align}
This inner product is sometimes called the \emph{Symmetric Logarithmic Derivative (SLD)} inner product~\cite{Suzuki2016}.

We note that $\langle \cdot, \cdot\rangle_\rho$ is really an inner product only if $\rho$ is strictly positive definite. Nevertheless, by abuse of terminology, we still call it an inner product even if $\rho$ is not faithful, in which case we essentially use the fact that $\langle \cdot, \cdot\rangle_\rho$ is a non-negative bilinear form.  
We denote the induced norm by\footnote{This norm should not be confused with the $2$-norm $\|\cdot\|_2$ defined in Subsection~\ref{subsec:p-norm}.} 
$$\|X\|_{2, \rho} := \langle X, X\rangle_\rho^{\frac 12}.$$
Observe that $\|X\|_{2, \rho} = \|X^\dagger\|_{2, \rho}$.

We will need yet another notion to define our quantum Poincar\'e inequality. Given an operator $X$ acting on an $m$-mode bosonic system, we define its \emph{gradient vector} by
\[
	\partial X = \big([\ac_1, X], [\ac_1^\dagger, X], \dots , [\ac_m, X], [\ac_m^\dagger, X]\big).
\]
We think of the $j$-th pair of coordinates of this vector as the derivatives in the $j$-th direction.

\begin{definition}\label{def:Poincare-ineq}
Let $\rho$ be an $m$-mode quantum state. We say that $\rho$ satisfies the Poincar\'e inequality (with respect to the SLD inner product) with constant $\lambda_\rho>0$ if for any operator $X$ satisfying $\tr(\rho X)=0$ and $\|X\|_{2, \rho}<+\infty$ we have 
\[
	\lambda_\rho \|X\|_{2, \rho}^2\leq  \|\partial X\|_{2, \rho}^2,
\]
where $\|\partial X\|_{2, \rho}^2:= \sum_{j=1}^m \Big(\big\|[\ac_j, X]\big\|_{2, \rho}^2 + \big\|[\ac_j^\dagger, X]\big\|_{2, \rho}^2\Big)$. Here, by convention, if the right hand side is not finite, we interpret it as $+\infty$ in which case the above inequality holds true. 
\end{definition}

Let us assume that $\rho$ is strictly positive definite. In this case, the space of operators $X$ satisfying $\|X\|_{2, \rho}<+\infty$ is a Hilbert space which we denote by $\mathbf L_2(\rho)$.\footnote{More precisely, $\mathbf L_2(\rho)$  is the completion of the space of bounded operators under the norm $\|\cdot\|_{2, \rho}$.} Then, the gradient map $\partial$ can be considered as an \emph{unbounded} operator mapping $\mathbf L_2(\rho)$ to $\mathbf L_2(\rho)^{2m}$. Letting $\partial^*$ be its adjoint, the Poincar\'e constant $\lambda_\rho$ can be thought of as the spectral gap of $\partial^*\partial$. The point is  that $\partial\, \mathbb I=0$ and the condition $\tr(\rho X)=0$ is equivalent to $\langle X, \mathbb I\rangle_\rho=0$. Thus, $\lambda_\rho$ equals the spectral gap of $\partial^*\partial$ above the zero eigenvalue with the corresponding eigenvector $\mathbb I$.

In Appendix~\ref{app:PoincareInequality} we explore some properties of quantum Poincar\'e inequalities. In particular, we show that if $\lambda_\rho>0$ for a state $\rho$, then $\lambda_{U\rho U^\dagger}>0$ for all Gaussian unitaries $U$. We also show that $\lambda_{\rho \boxplus_\eta \sigma}\geq \min\{\lambda_\rho, \lambda_\sigma\}$. Moreover, we show that the Poincar\'e inequality implies finiteness of all moments. Furthermore, we show that the Poincar\'e constant is positive for all Gaussian states that are strictly positive definite.

\subsection{Functional calculus for bivariate functions} \label{subsec:pi-representation} 

For a quantum state $\rho$ and continuous function $g(x, y): [0,1]\to \mathbb C$  we define the superoperator $\PiInner{\rho}{g}(\cdot)$ by
\begin{equation}\label{eq:def_Pi}
	\PiInner{\rho}{g}(X) := g(L_\rho, R_\rho)(X),
\end{equation}
where $L_\rho, R_\rho$ are the left and right multiplications by $\rho$, respectively, i.e., $L_\rho(X) = \rho X$ and $R_\rho(X) = X\rho$. We note that $L_\rho, R_\rho$ commute, so $g(L_\rho, R_\rho)$ is well defined; see~\cite{bardet2017estimating} for more details.
  
From the definition we clearly have
$$\PiInner{\rho}{g_1}\circ \PiInner{\rho}{g_2} = \PiInner{\rho}{g_1g_2}.$$
Moreover, if $g(x, y)$ is real-valued, then $\PiInner{\rho}{g}$ is self-adjoint with respect to the Hilbert-Schmidt inner product, meaning that 
$$\tr\Big( X^\dagger \PiInner{\rho}{g}(Y)  \Big) = \tr\Big( \big(\PiInner{\rho}{g}(X)\big)^\dagger Y  \Big).$$
This is simply because $L_\rho, R_\rho$ are self-adjoint with respect to the Hilbert-Schmidt inner product.

We will use the following proposition from~\cite{bardet2017estimating}.

\begin{proposition}\label{prop:pi-representation}
Suppose that $g(x, y), h(x, y)$ are two real valued functions satisfying $g(x, y)\geq h(x, y)$. Then, for any $X$ we have
$$\tr\Big(  X^\dagger \PiInner{\rho}{g}(X)  \Big) \geq \tr\Big(  X^\dagger \PiInner{\rho}{h}(X)  \Big). $$
\end{proposition}

\section{Convergence rate for trace distance} \label{SecTraceDistanceProof}

This section contains the proof of our main result regarding the convergence rate of the quantum CLT in terms of the trace distance. 

\begin{theorem} \label{TraceMainTheoremM}
	Let $\rho$ be an $m$-mode quantum state. Suppose that $\rho$ is centered (i.e., $\mathbf{d}(\rho)=0$) and has a finite moment of order $\kappa=\max\{2m, 3\}$. Let $\rho_{G}$ be the Gaussification of $\rho$ as in Definition~\ref{def:Gaussification}. Then, we have
	\begin{equation}\label{eq:TraceMainTheoremM}
		\big\| \rho^{\boxplus n} - \rho_{G}\big\|_{1} = \mathcal{O}\Big(\frac{1}{\sqrt{n}}\Big).
	\end{equation}
	
\end{theorem}

\medskip
Recall that the characterization function is very well behaved under convolution:
\[
	\chi_{\rho^{\boxplus n}}(z) = \chi_{\rho}\Big(\frac {z}{\sqrt n}\Big)^{n}.
\]
Using this equation it is not hard to verify the convergence of $\chi_{\rho^{\boxplus n}}(z)$ to $\chi_{\rho_G}(z)$. Nevertheless, this convergence does not directly imply anything about the rate of convergence in trace distance. However, if we replace the trace distance with the $2$-norm, then we can use the Plancherel identity~\eqref{PlancherelEq} to prove a bound on the rate of convergence in terms of characteristic functions. This idea is pursued in~\cite{CRLimitTheorem} and is followed by the use of the so called gentle measurement lemma to obtain a bound on the convergence rate in trace distance in terms of that of the $2$-norm. Here, we essentially follow the same  line of ideas to use the Plancherel identity, but take a different approach to relate trace distance to $2$-norm. To this end, we use a simple H\"older's inequality to derive a quantum generalization of an inequality that is a starting point of several CLTs in the classical case. In fact, we use the following lemma that is a quantum generalization of~\cite[Lemma 11.6]{BhRa}.

\begin{lemma}\label{lem:BoundTChi}
Let $T$ be an operator such that $A^\dagger TA$ is trace class, where $A =  \ac_1 \cdots \ac_m$. 
Then, we have
\begin{equation} \label{BoundTChi}
	{\| T \|}_{1}^{2} \leq \Big(\frac{\pi^2}{6}\Big)^{m} \int \big|\chi_{A^{\dagger} T A}(z)\big|^{2} \dd^{2m} z,
\end{equation}
\end{lemma}

\begin{proof}
We first note that $AA^\dagger = \ac_1\ac_1^\dagger \cdots \ac_m\ac_m^\dagger = (\ac^\dagger_1\ac_1+1) \cdots (\ac^\dagger_m\ac_m+1)$ is invertible. Then, by H\"older's inequality we have
\begin{align*}
	\| T \|_{1} &= \tr\big(|T|\big)\\
	& = \tr\big(\big| (AA^\dagger)^{-\frac{1}2} (AA^\dagger)^{\frac 12} T (AA^\dagger)^{\frac 1 2} (AA^\dagger)^{-\frac{1}2} \big|\big) \\
	& \leq \big\| (AA^\dagger)^{-\frac12} \big\|_{4}^{2} \cdot \big\| (AA^\dagger)^{\frac 12} T (AA^\dagger)^{\frac12} \big\|_{2}.
\end{align*}
Next, computing the trace in the Fock basis we obtain
\begin{align*}
	\big\| (AA^\dagger)^{-\frac 12} \big\|_{4}^{4} = \tr\big((AA^\dagger)^{-2}\big) = \sum_{n_1, \dots, n_m = 1}^{\infty} \frac{1}{n_{1}^2 \dots n_{m}^2} = \Big(\frac{\pi^2}{6}\Big)^{m}.
\end{align*}
Also, Plancherel's identity~\eqref{PlancherelEq} yields
\begin{align*}
\big\| (AA^\dagger)^{\frac 12} T (AA^\dagger)^{\frac12} \big\|_{2}^2 & =  \tr\big(  AA^\dagger T AA^\dagger T^\dagger \big) 
	= \big\| A^{\dagger} T A \big\|_{2}^{2}= \int \big|\chi_{A^{\dagger} T A}(z)\big|^{2} \dd^{2m} z.
\end{align*}
Putting these together the desired inequality is implied. 

\end{proof}

We need yet another lemma to prove Theorem~\ref{TraceMainTheoremM}.

\begin{lemma} \label{BRMethod}
Let $\rho$ be an $m$-mode quantum state that is centered and has finite moment of order $\kappa \geq 3$. Then, there exist $\epsilon>0$, and a polynomial $q_{r}(\cdot)$ for any $1 \leq r \leq \kappa -2$, such that the following holds: for all $z \in \mathbb{C}^{m}$ satisfying $\|z\|\leq \epsilon\sqrt n$ and all $\alpha, \beta\in \mathbb Z_+^m$ satisfying $|\alpha| + |\beta| \leq \kappa $ we have
\[
	\Big|\partial_{z}^\alpha  \partial_{\bar z}^\beta \Big(\chi_{\rho^{\boxplus n}}(z) - \chi_{\rho_{G}}(z)\Big)\Big| \leq e^{- \frac{\nu_{{\min}(\rho)}-1}{4} \|z \|^2} \sum_{r=1}^{\kappa-2} n^{-\frac r2} q_{r}(\| z\|),
\]
where $\nu_{\min}(\rho)\geq 1$ is the minimum eigenvalue of the covariance matrix of $\rho$. 
\end{lemma}

\begin{proof}
This lemma is standard in the theory of CLTs in the classical case and can be proved by similar computations as in the classical case. The whole point is that $\chi_{\rho^{\boxplus n}}(z) = \chi_{\rho}\Big(\frac {z}{\sqrt n}\Big)^{n}$ and  the behavior of the characteristic function under quantum convolution is the same as that in the classical case. In fact, this lemma can be proved even from the existing literature as follows. First, suppose that the Wigner function of $\rho$ is non-negative. This means that we can think of $\chi_{\rho}(z)$ as the characteristic function of a classical distribution, the Wigner function~\cite{CRLimitTheorem}. Moreover, using~\eqref{eq:moments-derivatives} and the fact that the moments of order $\kappa$ of this classical distribution can be written as the derivatives of the characteristic function, moments of order $\kappa$ of this classical distribution are also finite. Then, using~\cite[Theorem 9.9]{BhRa} and~~\cite[Lemma 9.4]{BhRa}, we have

\[
	\Big|\partial_{z}^\alpha  \partial_{\bar z}^\beta \Big(\chi_{\rho^{\boxplus n}}(z) - \chi_{\rho_{G}}(z)\Big)\Big| \leq e^{- \frac{\nu_{\text{min}(\rho)}}{4} \|z \|^2} \sum_{r=1}^{\kappa-2} n^{-\frac r2} q_{r, \rho}(\| z\|),
\]
for some polynomials $q_{r, \rho}(\cdot)$. 

In case the Wigner function of $\rho$ takes negative values, we consider the quantum state $\sigma = \rho \boxplus \ketbra{0}{0}$. By~\cite[Lemma 16]{CRLimitTheorem}, the Wigner function of $\sigma$ is non-negative and by the above argument we have 
\begin{equation}\label{eq:Edge_Vaccuum}
\Big|\partial_{z}^\alpha  \partial_{\bar z}^\beta \Big(\chi_{\sigma^{\boxplus n}}(z) - \chi_{\sigma_{G}}(z)\Big)\Big| \leq e^{- \frac{\nu_{\text{min}}(\sigma)}{4} \|z \|^2} \sum_{r=1}^{\kappa-2} n^{-\frac r2} q_{r, \sigma}(\| z\|).
\end{equation}
We note that by~\eqref{eq:tau-chi-W} and~\eqref{eq:convolution-char-function} we can write
\[
	\chi_{\rho^{\boxplus n}}(z) = \chi_{\sigma^{\boxplus n}}(\sqrt 2 z) e^{\frac 12 \|z\|^2}, \quad \chi_{\rho_G}(z) = \chi_{\sigma_G}(\sqrt 2 z) e^{\frac 12 \|z\|^2}.
\]
Now according to~\eqref{eq:Edge_Vaccuum} and a simple triangle inequality, for some polynomials $q_r(\cdot)$,  we can write
\begin{align*}
	\Big|\partial_{z}^\alpha  \partial_{\bar z}^\beta \Big(\chi_{\rho^{\boxplus n}}(z) - \chi_{\rho_{G}}(z)\Big)\Big| 
	&\leq   e^{- \frac{\nu_{\text{min}}(\sigma) - 1}{2} \|z \|^2} \sum_{r=1}^{\kappa-2} n^{-\frac r2} q_{r}(\| z\|).
\end{align*}
The last point here is that $\nu_{\min}(\sigma) = \frac 12(\nu_{\min}(\rho) + 1)$.

\end{proof}

\subsection{Proof of Theorem~\ref{TraceMainTheoremM}} 

By applying an appropriate Gaussian unitary and using Proposition~\ref{Will}, we assume that $\rho$ is in Williamson's form and the covariance matrix of $\rho$ is diagonal as in~\eqref{eq:Williamson-diag}. This can be done with no loss of generality since by~\eqref{UnitaryAndConv} Gaussian unitaries commute with quantum convolution. 

If $\nu_j=1$ for some $j$, as discussed in Subsection~\ref{subsec:Gaussian-states}, this means that in $\rho$ the $j$-th mode is decoupled from other modes and is in the vacuum state. Moreover, the vacuum state is a centered Gaussian state and remains unchanged under convolution with itself. Thus, in order to prove the theorem, we may ignore the $j$-th mode. As a conclusion, we assume that $\nu_j>1$ for all $j$, or equivalently $\nu_{\min}(\rho)>1$.

Now we apply Lemma~\ref{lem:BoundTChi} on $$T= \rho^{\boxplus n} - \rho_G.$$ 
Then, for a given $\epsilon>0$ we break the integral in~\eqref{BoundTChi} in two parts:
\begin{align}
	 \Big(\frac{6}{\pi^2}\Big)^{m}  \| \rho^{\boxplus n} - \rho_G \|_{1}^{2} & \leq \int \big|\chi_{A^{\dagger} T A}(z)\big|^{2} \dd^{2m} z \nonumber \\
	 & = \int_{\|z\| \leq \epsilon \sqrt{n} } \big|\chi_{A^{\dagger} T A}(z)\big|^{2} \dd^{2m} z + \int_{\|z\| > \epsilon \sqrt{n} } \big|\chi_{A^{\dagger} T A}(z)\big|^{2} \dd^{2m} z. \label{splitInt}
\end{align}
We show that the first term on the right hand side of~\eqref{splitInt} is of the order $\mathcal{O}(\frac{1}{n})$ and the second term is exponentially small. This gives the desired result. 

We start with the first term.  A straightforward application of the BCH formula implies that the displacement operator can be written as~\cite{Serafini}
\[ 
	D_z = \prod_{j=1}^m e^{- \frac 12|z_j|^{2}} e^{z_j \ac_j^{\dagger}} e^{-\bar z_j \ac_j} = \prod_{j=1}^m e^{\frac12 {|z_j|}^{2}} e^{-\bar z_j \ac_j} e^{z_j \ac_j^{\dagger}}.
\]
Computing the derivatives, we obtain
\begin{align}\label{eq:displacement-derivative}
	\partial_{z_j} D_z = \frac{1}{2} \bar z_{j} D_z + D_z \ac_j^{\dagger},
\qquad	\partial_{\bar z_{j}} D_z  = \frac{1}{2} z_j D_z - \ac_j D_z.
\end{align}
Then, using these in the definition of the characteristic function we arrive at
\[
	\partial_{z_j} \chi_{T}(z) = \frac{1}{2} \bar z_{j} \chi_{T}(z) + \chi_{\ac_j^{\dagger}T}(z), 
	\qquad
	\partial_{\bar z_{j}} \chi_{T}(z) = \frac{1}{2} z_j \chi_{T}(z) - \chi_{T \ac_j}(z).
\]
Equivalently, we have
\[
	\chi_{\ac_j^{\dagger}T}(z) = \Big(\partial_{z_j} - \frac{1}{2} \bar z_{j}\Big) \chi_{T}(z),
\qquad
	\chi_{T \ac_j}(z) = \Big(- \partial_{\bar z_{j}} + \frac{1}{2} z_j \Big) \chi_{T}(z).
\]
Therefore, for some constants $c_{\alpha, \alpha', \beta, \beta'}$ we have 
\begin{align*}
	\chi_{A^{\dagger} T A} (z) 
	&=  \prod_{j=1}^{m} \Big(\partial_{z_j} - \frac{1}{2} \bar z_{j}\Big) \Big(- \partial_{\bar z_{j}} + \frac{1}{2} z_j \Big)  \chi_{T}(z) \\
	&= \sum_{\alpha, \alpha', \beta, \beta' } c_{\alpha, \alpha', \beta, \beta'} z^{\alpha'}\bar z^{\beta'}\partial_z^\alpha\partial_{\bar z}^\beta \chi_{T}(z),
\end{align*}
where the sum runs over $\alpha, \alpha', \beta, \beta'\in \{0,1\}^m$ such that $|\alpha|+|\alpha'|+|\beta|+|\beta'|<2m$.
Thus, there exists a constant $C>0$ such that  
\begin{equation} \label{ChiDerivative}
	\big|\chi_{A^{\dagger} T A} (z)\big|^{2} \leq C \sum_{\alpha, \alpha', \beta, \beta'} \prod_{j=1}^{m} |z_{j}|^{(\alpha'_j+\beta'_j)} \big|\partial_z^\alpha\partial_{\bar z}^\beta \chi_{T}(z) \big|^{2}.
\end{equation}
Next, applying Lemma~\ref{BRMethod} with $\kappa=\max\{3, 2m\}$ for each term on the right hand side, we find that there is some $\epsilon>0$ such that for some polynomials $q_r^\prime(\cdot)$ and some constants $\xi^\prime>0$, $C'>0$, the following inequality holds for all $\|z\| \leq \epsilon \sqrt{n}$:
\[
	\big|\chi_{A^{\dagger} T A} (z)\big|^{2} \leq C' \sum_{r=2}^{2(\kappa-2)} n^{-\frac r2} q_{r}^\prime(\| z\|) e^{- \xi^\prime \|z \|^2}.
\]
Note that here we use $\nu_{\min}(\rho)>1$.
Then, taking the integral of both sides and using
\[
	\int q_{r}^\prime(\| z\|) e^{- \xi^\prime \|z \|^2}  \\d^{2m} z < +\infty,
\]
imply that 
$$\int_{\|z\| \leq \epsilon \sqrt{n} } \big|\chi_{A^{\dagger} T A}(z)\big|^{2} \dd^{2m} z= \mathcal{O}\Big(\frac{1}{n}\Big).$$

We now handle the second term in~\eqref{splitInt}. We have
$$	\chi_{A^{\dagger} T A}(z) = \chi_{A^{\dagger} \rho^{\boxplus n} A}(z) - \chi_{A^{\dagger} \rho_{G} A}(z).$$
Inequality $(x+y)^2\leq 2(x^2+y^2)$ yields
\begin{align*}
	\int_{\|z\| > \epsilon \sqrt{n} } \big|\chi_{A^{\dagger} T A}(z)\big|^{2} \dd^{2m} z \leq 2 \bigg{(} \int_{\|z\| > \epsilon \sqrt{n} } \big|\chi_{A^{\dagger} \rho^{\boxplus n} A}(z)\big|^{2} \dd^{2m} z  + \int_{\|z\| > \epsilon \sqrt{n} } \big|\chi_{A^{\dagger} \rho_{G} A}(z)\big|^{2} \dd^{2m} z \bigg{)}.
\end{align*}
We will show that the first term on the right-hand side is exponentially small. Then, it is implied that the second term is also exponentially small because we have $\rho_{G} = \rho_{G}^{\boxplus n}$. 
To this end, we use 
\begin{align}\label{eq:tr-A-rho-A-D} 
\tr\Big( A^{\dagger}  \rho^{\boxplus n} A D_z\Big) = \tr\Big(\widetilde A^\dagger \rho^{\otimes n} \widetilde A D_{\frac {z}{\sqrt n}}^{\otimes n}\Big),
\end{align}
where 
$$\widetilde A = \prod_{j=1}^m \frac{\ac_{j,1}+ \cdots +\ac_{j, n} }{\sqrt n}.$$
This equation can be proven by computing the derivatives of $\chi_{\rho^{\boxplus n}}(z) = \tr\big(\rho^{\boxplus n}D_{z}\big) = \tr\Big(\rho^{\otimes n} D_{\frac{z}{\sqrt n}}^{\otimes n}\Big)$ in terms of those of $D_{z}$ and $D_{\frac{z}{\sqrt n}}^{\otimes n}$ as done above; see also Proposition~\ref{Basictilde} below. Then expanding the right hand side of~\eqref{eq:tr-A-rho-A-D}, we find that
$$\tr\Big( A^{\dagger}  \rho^{\boxplus n} A D_z\Big) =  \frac{1}{n^m} \tr\big(\rho D_{\frac{z}{\sqrt n}}\big)^{n-2m}P(z)= \frac{1}{n^m} \chi_\rho\big(z/\sqrt n\big)^{n-2m} P(z),$$
where $P(z)$ equals
$$P(z) = \sum_{1 \leq k_1,\ell_1, \dots, k_m, \ell_m \leq n}   \tr\Big(\ac_{1, k_1}^\dagger \cdots \ac_{m, k_m}^\dagger \rho^{\otimes 2m} \ac_{1, \ell_1} \cdots \ac_{m, \ell_m} D_{\frac{z}{\sqrt z}}^{\otimes 2m} \Big).$$
Now, using~\eqref{SupChar}, there is $\delta>0$ such that $\big|\chi_\rho\big(z/\sqrt n\big)\big|\leq 1-\delta$ for all $\|z\|\geq \epsilon\sqrt n$. Putting these together and applying the Plancherel identity~\eqref{PlancherelEq} yield
\begin{align*}
 \int_{\|z\| > \epsilon \sqrt{n} } &\big| \chi_{A^\dagger \rho^{\boxplus n} A} (z)\big|^{2} \dd^{2m} z  \\
 &\leq  \frac{1}{n^{2m}} (1-\delta)^{2(n-2m)}  \int_{\|z\| > \epsilon \sqrt{n} }   \big| P(z)\big|^2 \dd^{2m} z\\
&\leq \frac{1}{n^{2m}} (1-\delta)^{2(n-2m)}  \int   \big| P(z)\big|^2 \dd^{2m} z\\
&\leq (1-\delta)^{2(n-2m)}  \sum_{k_1,\ell_1, \dots, k_m, \ell_m}   \int  \Big| \tr\Big(\ac_{1, k_1}^\dagger \cdots \ac_{m, k_m}^\dagger \rho^{\otimes 2m} \ac_{1, \ell_1} \cdots \ac_{m, \ell_m} D_{\frac{z}{\sqrt z}}^{\otimes 2m} \Big) \Big|^2 \dd^{2m} z\\
&\leq (1-\delta)^{2(n-2m)}  \sum_{k_1,\ell_1, \dots, k_m, \ell_m}     \Big\| \ac_{1, k_1}^\dagger \cdots \ac_{m, k_m}^\dagger \rho^{\otimes 2m} \ac_{1, \ell_1} \cdots \ac_{m, \ell_m} \Big\|_2^2.
\end{align*}
Next, we observe that 
\begin{align*}
 \Big\| \ac_{1, k_1}^\dagger \cdots \ac_{m, k_m}^\dagger \rho^{\otimes 2m} \ac_{1, \ell_1} \cdots \ac_{m, \ell_m} \Big\|_2^2 & \leq  \Big\| \ac_{1, k_1}^\dagger \cdots \ac_{m, k_m}^\dagger \rho^{\otimes 2m} \ac_{1, \ell_1} \cdots \ac_{m, \ell_m} \Big\|_1^2 <+\infty,
\end{align*}
since by assumption $\rho$ has a finite moment of order $\kappa=\max \{3, 2m\}$. Using this in the previous inequality, and observing that there are $n^{2m}$ terms  on the right hand side, we find that the left hand side is exponentially small. 
This concludes the proof of Theorem~\ref{TraceMainTheoremM}.

\subsection{Optimally of the convergence rate}\label{subsec:opt-conv-rate}

In this subsection we show that the convergence rate proven in Theorem~\ref{TraceMainTheoremM} is tight.
Let  $\rho=\ketbra{v}{v}$ where
$$\ket v= \frac{1}{\sqrt 2} \big(\ket 0+ \ket 3\big),$$
is a superposition of Fock states $\ket 0, \ket 3$. 
This example is already considered in~\cite{CRLimitTheorem} where it is shown that 
$$\chi_\rho(z) = e^{-\frac 12 |z|^2} \Big(1- \frac 32|z|^2 + \frac{1}{2\sqrt 6}(z^3- \bar z^3) + \frac 34|z|^4 - \frac{1}{12} |z|^6   \Big).$$
Therefore, for any $z\in \mathbb C$ we have
\begin{align*}
\chi_{\rho^{\boxplus n}}(z) & = \chi_\rho\bigg(\frac{z}{\sqrt n}\bigg)^n\nonumber\\
& = e^{-\frac 12 |z|^2} \bigg[1- \frac{3}{2n}|z|^2 + \frac{1}{2\sqrt 6 n^{3/2}}(z^3- \bar z^3) + \frac{3}{4 n^2}|z|^4 - \frac{1}{12 n^3}|z|^6   \bigg]^n\nonumber\\
& = e^{-\frac 12 |z|^2} \bigg[ \bigg(1- \frac{3}{2n}|z|^2\bigg)^n +  n \frac{1}{2\sqrt 6 n^{3/2}}(z^3- \bar z^3)\bigg(1- \frac{3}{2n}|z|^2\bigg)^{n-1} +  \mathcal O\Big(\frac{1}{n}\Big) \bigg]\nonumber\\
& = e^{-2 |z|^2} \bigg[ 1+ \frac{1}{2\sqrt 6 n^{1/2}}(z^3- \bar z^3) +  \mathcal O\Big(\frac{1}{n}\Big) \bigg]. 
\end{align*} 
As a result, $\rho^{\boxplus n}$ convergence to a Gaussian state (in fact, a thermal state) with characteristic function $\chi_{\rho_G}(z) = e^{-2|z|^2}$.  We also conclude that 
\begin{align}\label{eq:opt-conv-rate}
\big|\chi_{\rho^{\boxplus n}}(z) - \chi_{\rho_G}(z)\big|\geq \frac{1}{100 \sqrt n}  |z^3- \bar z^3|,
\end{align}
for any $|z|\leq 1$ and sufficiently large $n$. On the other hand, by the definition of the characteristic function and using the fact that displacement operators are unitary, for sufficiently large $n$ we have 
\begin{align*}
\|\rho^{\boxplus n} - \rho_G\|_1 & = \sup_{\|X\|\leq 1} \big|\tr\big(\rho^{\boxplus n} X\big)  - \tr(\rho_G X) \big|\\
& \geq  \sup_{z} \big|\tr\big(\rho^{\boxplus n} D_z\big)  - \tr(\rho_G D_z) \big|\\
& =  \sup_{z} \big|\chi_{\rho^{\boxplus n}} (z)  - \chi_{\rho_G}(z) \big|\\
& \geq \frac{1}{50\sqrt n}.
\end{align*}
Here, the last inequality follows from~\eqref{eq:opt-conv-rate} for the choice of $z=i$.
We note that all the moments of $\rho$ are finite. Thus, the convergence rate of $\mathcal O(1/\sqrt n)$ proven in Theorem~\ref{TraceMainTheoremM} is indeed tight even assuming that all the moments of $\rho$ are finite.

%****************************************************************

\section{Convergence rate for relative entropy} \label{SimplifiedFisherInformation}

In this section we state our main result regarding the convergence rate of the entropic version of quantum CLT.

\begin{theorem} \label{MainTheorem}
	Let $\rho$ be a centered $m$-mode quantum state.  Suppose that $\rho$ satisfies the quantum Poincar\'e inequality as in Definition~\ref{def:Poincare-ineq}. Let $\rho_{G}$ be the Gaussification of $\rho$ as in Definition~\ref{def:Gaussification}. Then, we have
	\begin{equation}
		D(\rho^{\boxplus n} \| \rho_{G}) = \mathcal{O}\Big(\frac{1}{n}\Big).
	\end{equation}
\end{theorem}	

We note that by Pinsker's inequality this theorem implies
	\begin{equation}
		{\| \rho^{\boxplus n} - \rho_{G}\|}_{1} = \mathcal{O}(\frac{1}{\sqrt{n}}),
	\end{equation}
that is the same bound as the bound of Theorem~\ref{TraceMainTheoremM}, but with a much stronger assumption. Indeed, in Theorem~\ref{TraceMainTheoremM} we assume the finiteness of moments up to order $\max\{3, 2m\}$, yet it can be verified that $\lambda_\rho>0$ implies finiteness of all moments; see Appendix~\ref{app:PoincareInequality}.

Our proof of this theorem follows similar lines as in the proof of Johnson and Barron~\cite{JB} in the classical case. The idea is to express relative entropy in terms of Fisher information, and then establish a bound on the convergence rate of Fisher information. In fact, by \emph{de Bruijn's identity,} for independent random variables $X, Z$ where $Z$ is Gaussian with the same first and second moments as those of $X$,  we have
\[
	D(X \| Z) = \int_{0}^{1} \frac{J(\sqrt{t} X + \sqrt{1-t} Z)}{2t} \dd t.
\]
Here, $J(\cdot)$ denotes the \emph{Fisher information distance} that is a scaled and shifted version of the Fisher information $I(\cdot)$; for a random variable $X$ with density $f$, the Fisher information is defined by $I(X) = \int \frac{f'^2}{f}$ and the Fisher information distance equals $J(X) = \text{Var}(X) I(X) -1$. The above identity is the starting point of~\cite{JB} as well as \cite{ABBN2} based  on which the proof of the convergence rate of $\mathcal{O}(\frac{1}{n})$ for relative entropy is reduced to that of Fisher information distance. 

Unlike the classical case, there are various forms of quantum Fisher information~\cite{PETZ1996, LesniewskiRuskai99}. One of the most relevant such a quantum Fisher information is
\begin{equation} \label{HardFisherInformation}
	I_{\text{KMB}}(\rho) = \sum_{j =1}^{m} \tr\big(\rho \big[\ac_{j}^\dagger, [\ac_{j}, \log \rho]\big]\big),
\end{equation}
that is called the Kubo-Mori-Bogoliubov Fisher information~\cite{Hayashi12002Fisher}. This quantum Fisher information 
satisfies the quantum de Bruijn's identity~\cite[Theorem V.1]{KS}, so is a natural candidate for generalizing the approaches of~\cite{JB, ABBN2}. Nevertheless,~\eqref{HardFisherInformation} is hard to work with and does not satisfy some of the desirable properties of the classical Fisher information. Thus, we proceed differently.   

The Fisher information of a random variable $X$ with density function $f$ can be thought of as the $2$-norm of the \emph{score function} $\frac{f'}{f}$ under the same reference measure: $I(X)=\int \frac{f'^2}{f} = \int f \big(\frac{f'}{f}\big)^2$. The same analogy holds for the quantum Fisher information~\eqref{HardFisherInformation}. Indeed, as pointed out, e.g., in~\cite{CarlenMass2}, it holds that
\begin{align}\label{eq:q-score-functioin-hard}
	 [\ac_j, \log \rho] = \PiInner{\rho}{\widetilde \log}([\ac_j, \rho]) = \int_{0}^{\infty} \frac{1}{r+\rho} [\ac_j, \rho] \frac{1}{r+ \rho} \dd r,
\end{align}
where $\widetilde \log(x, y) = \frac{\log x-\log y}{x-y} = \int_0^{\infty} \frac{1}{(r+x)(r+y)}\dd r$.
Then, a simple algebra verifies that
\begin{align*}
I_{\text{KMB}}(\rho) & = \sum_{j=1}^m   \int_{0}^{\infty}  \tr\Big(  [\ac_j, \rho]^\dagger \frac{1}{r+\rho} [\ac_j, \rho] \frac{1}{r+ \rho}  \Big) \dd r.
\end{align*}
As mentioned in Subsection~\ref{subsec:Poincare}, we may think of $[\ac_j, \rho], [\ac_j^\dagger, \rho]$ as the derivatives of $\rho$ in the $j$-th direction. On the other hand, $\int_{0}^{\infty} \frac{1}{(r+\rho)^2}\dd r = \frac{1}{\rho}$. Thus, the above equation resembles the definition of Fisher information $I(X)=\int \frac{f'^2}{f}$ in the classical case. Even more, we may take~\eqref{eq:q-score-functioin-hard} as the \emph{quantum score operator} in which case the quantum Fisher information~\eqref{HardFisherInformation} can be expressed in terms of the $2$-norm of these operators with respect to the Kubo-Mori-Bogoliubov inner product defined in~\eqref{eq:BKM-inner-prod}.\footnote{The point is that $\int_0^1 x^{1-s}y^s \dd s = \frac{x-y}{\log x-\log y} = \frac{1}{\widetilde \log(x, y)}$.}

\subsection{SLD quantum Fisher information}
Recall that in Subsection~\ref{subsec:Poincare} we considered the SLD inner product  
\begin{equation*} 
	{\langle X, Y\rangle}_{ \rho} = \frac{1}{2}\Big(\tr\big(\rho X^{\dagger} Y\big) + \tr\big(X^{\dagger} \rho Y\big)\Big).
\end{equation*}
Using the functional calculus discussed in Subsection~\ref{subsec:pi-representation}, this inner product can be written as
\[
	{\langle X, Y\rangle}_{ \rho} = \tr\big(\PiInner{\rho}{\psi}(X^{\dagger}) Y\big),
\]
where hereafter the function $\psi$ is fixed to be
\[
	\psi(x,y) = \frac{x+y}{2}.
\]
Now to define the quantum score operator corresponding to this inner product, it is more convenient to assume that $\rho$ is in Williamson's form according to Proposition~\ref{Will}. In this case, $\rho_G=\tau_1\otimes \cdots \otimes \tau_m$ is a tensor product of single-mode thermal states as in~\eqref{eq:thermal}. For such a state $\rho$ we define the \emph{SLD score operator} in the $j$-th direction by
\begin{equation} \label{NewMSFunction}
	S_{\rho, j} := \PiInner{\rho}{\phi}\big([\ac_j, \rho]\big),
\end{equation}
where
\[
	\phi(x,y) := \frac{1}{\psi(x,y)}= \frac{2}{x+y}. 
\]
Here, we think of the superoperator $\PiInner{\rho}{\phi}$ as a non-commutative way of division by $\rho$. The point is that $\PiInner{\rho}{\psi}$ is a non-commutative multiplication by $\rho$ and since $\phi=1/\psi$, the superoperator $\PiInner{\rho}{\phi}$ is a non-commutative division by $\rho$. Thus, we can think of $S_{\rho, j}=\PiInner{\rho}{\phi}\big([\ac_j, \rho]\big)$ as the derivative of $\rho$ in the $j$-th direction divided by $\rho$, which resembles the form of score function $f'/f$ in the classical case.

We define the \emph{SLD Fisher information} $ I(\rho) = I_{\text{SLD}}(\rho)$ by
$$I(\rho) = \sum_{j=1}^m I_j(\rho),$$
where 
\begin{align}
	 I_j(\rho): = \| S_{\rho, j} \|_{2, \rho}^{2}
	&= \tr\big(\PiInner{\rho}{\psi}(S_{\rho, j}^{\dagger}) S_{\rho, j}\big) \nonumber \\
	&= \tr\Big(\PiInner{\rho}{\psi}(\PiInner{\rho}{\phi}({[\ac_j, \rho]}^{\dagger})) \PiInner{\rho}{\phi}([\ac_j, \rho])\Big) \nonumber \\
	&= \tr\big(\PiInner{\rho}{\psi \phi}({[\ac_j, \rho]}^{\dagger}) \PiInner{\rho}{\phi}([\ac_j, \rho])\big) \nonumber \\
	&= \tr\big({[\ac_j, \rho]}^{\dagger} \PiInner{\rho}{\phi}([\ac_j, \rho])\big). \label{eq:I_j-phi-com}
\end{align}
We emphasize that we define the SLD Fisher information for an arbitrary state, by first bringing it to Williamson's form. Since Williamson's form is unique up to some passive Gaussian unitaries~\cite{Serafini}, and such unitaries leave $\sum_j I_j(\rho)$ invariant, $I(\rho)$ is well-defined at least when $\rho$ is faithful. To clarify the validity of the definition in case $\rho$ is not faithful, we note that
\begin{align}
I_j(\rho)  = \tr\Big(\rho \bfa_j^\dagger \PiInner{\rho}{\phi}(\bfa_j\rho) - \rho \bfa_j^\dagger \PiInner{\rho}{\phi}(\rho \bfa_j) - \bfa_j^\dagger \rho \PiInner{\rho}{\phi}(\bfa_j\rho) + \bfa_j^\dagger \rho \PiInner{\rho}{\phi}(\rho\bfa_j)   \Big)  
 = \tr\Big(\bfa_j^\dagger \PiInner{\rho}{\zeta}(\bfa_j)\Big) \label{eq:Fisher_Zeta},
\end{align}
where 
$$\zeta(x, y)= (y^2-2xy+x^2)\phi(x, y) = 2\frac{(x-y)^2}{x+y} = 2(x+y) - 8 \frac{xy}{x+y}.$$
Now the point is that $\zeta(x, y)$ is continuous even at $(x, y)=(0, 0)$ with $\zeta(0,0)=0$. Thus, $\PiInner{\rho}{\zeta}(\cdot)$ is well-defined even if $\rho$ has a non-trivial kernel.

Next, as the name suggests, we define the \emph{SLD Fisher information distance} $J(\rho)$ as the difference of $I(\rho)$  and  $I(\rho_G)$ where $\rho_G$ is the Gaussification of $\rho$. That is,
\begin{align*}
J(\rho) & = \sum_{j=1}^m   J_j(\rho),
\end{align*}
where
$$J_j(\rho) = I_j(\rho) - I_j(\rho_G).$$
It is shown in Appendix~\ref{ModifiedGaussian} that for $\rho_G=\tau_1\otimes \cdots \otimes \tau_m$ where $\tau_j$ is a single-mode thermal state satisfying 
$$\mu_j:=\|\ac_j\|_{2, \tau_j}^2 = \|\ac_j\|_{2, \rho_G}^2 = \|\ac_j\|_{2, \rho}^2 = \frac{1}{2}\nu_j = \frac{1+ e^{-\beta_j}}{2(1- e^{-\beta_j})},$$
we have $S_{\rho_G, j} = -\frac{1}{\mu_j} \ac_j$  which implies  $I(\tau_j) = \frac{1}{\mu_j}$.
Then, for a state $\rho$ in Williamson's form we have
$$J_j(\rho) =  I_j(\rho) - \frac{1}{\mu_j}.$$
This equation reveals that the SLD Fisher information distance resembles its classical counterpart (up to a factor of $\mu_j$). 

We can derive yet another equivalent expression for the Fisher information distance. To this end, we first compute
\begin{align}\label{eq:S-a-inner-prod}
\langle S_{\rho, j}, \ac_j\rangle_\rho  = \tr\big(  [\ac_j, \rho]^\dagger \ac_j \big) = \tr\big(   \rho [\ac_j^\dagger , \ac_j] \big) =-1.
\end{align}
Therefore, 
\begin{align}\label{eq:J-I-Will}
J_j(\rho)=I_j(\rho) - \frac{1}{\mu_j} = \big\| S_{\rho, j}\big\|_{2, \rho}^2 + \frac{1}{\mu_j^2} \big\|\ac_j  \big\|_{2, \rho}^{2} - \frac{2}{\mu_j} = \Big\| S_{\rho, j} + \frac{1}{\mu_j} \ac_j  \Big\|_{2, \rho}^{2}.
\end{align}
where as before $\mu_j = \|\ac_j\|_{2, \rho}^2 = \frac 12\big(\tr(\rho \ac_j^\dagger \ac_j) + \tr(\rho \ac_j\ac_j^\dagger)\big) = \tr(\rho\ac_j^\dagger \ac_j) + \frac 12>0$.

In the following lemma we establish an upper bound on the SLD Fisher information, even if $\rho$ is not faithful.

\begin{lemma}\label{lem:Fisher-inf-bounds}
Let $\rho$ be an $m$-mode state in Williamson's form. Then,
for any $1\leq j\leq m$ we have
\begin{align} \label{BoundSimplifiedFisherInformation}
\frac{1}{\mu_j} \leq I_j(\rho) \leq 4\mu_j.
\end{align}
\end{lemma}

\begin{proof}
The first inequality is implied by using $J_j(\rho)\geq 0$ in~\eqref{eq:J-I-Will}. This inequality can be understood as a Cram\'er-Rao bound for the SLD Fisher information~\cite{Suzuki2016}. For the proof of the second inequality, we note that in~\eqref{eq:Fisher_Zeta}, $\zeta(x, y)\leq 2(x+y)$ for $x, y\geq 0$. Thus, by Proposition~\ref{prop:pi-representation} we have
\begin{align*}
	I_j(\rho) \leq \tr\Big(\bfa_j^\dagger \PiInner{\rho}{2(x+y)}(\bfa_j)\Big) = 2\tr\big(\ac_j^\dagger \rho\ac_j \big)+ 2\tr\big(\ac_j^\dagger \ac_j \rho\big) = 4 \|\ac_j\|_{2, \rho}^2 = 4\mu_j.
\end{align*}

\end{proof}

As mentioned above, if $\rho$ is Gaussian, then $J(\rho)=0$. Thus, the SLD Fisher information distance can really be thought of as a measure of distance to Gaussian states. This is a justification for using $J(\rho)$ to prove our entropic CLT. 

\begin{theorem} \label{pseudoMainTheorem}
	Let $\rho$ be a centered $m$-mode quantum state that satisfies the quantum Poincar\'e inequality with constant $\lambda_{\rho}>0$. Then, 
	\begin{align*}
		J\big(\rho^{\boxplus n}\big)=\mathcal O\bigg(    \frac {I(\rho)^2}{\lambda_\rho n}\bigg).
	\end{align*}
\end{theorem}

We leave the proof of this theorem for Section~\ref{ConvergenceFID}.

\subsection{Quantum log-Sobolev inequalities}
The SLD Fisher information defined above is much easier to work with. Nevertheless, it lacks the quantum de Bruijn's identity. To overcome this difficulty we diverge from the method of~\cite{JB} and instead of the quantum de Bruijn's identity use quantum log-Sobolev inequalities.

\begin{theorem}\emph{\cite{Beigi-RahimiKeshari2023}} \label{thm:LSI} 
Let  $\tau_j = (1-e^{-\beta_j}) e^{-\beta_j \ac_j^\dagger \ac_j}$ be a thermal state with parameter $0<\beta_j<+\infty$ and let $\tau = \tau_1\otimes \cdots \otimes \tau_m$. Then, for an $m$-mode quantum state $\rho$ we have
\begin{align*}
\alpha D(\rho\| \tau) & \leq  \sum_{j=1}^m \bigg( e^{\beta_j /2} \tr(\rho\ac_j^\dagger \ac_j)  +e^{-\beta_j/2}   \tr(\rho\ac_j \ac_j^\dagger)  - 2\tr\big( \rho^{\frac1 2}  \ac_j^\dagger  \rho^{\frac12} \ac_j    \big)   \bigg)\\
& =  \sum_{j=1}^m \tr\Big( \Big| e^{\beta_j/4} \ac_j \sqrt \rho - e^{-\beta_j/4} \sqrt \rho \ac_j   \Big|^2  \Big),
\end{align*}
where $\alpha^{-1}=  \frac{2+\log(2m+1)}{\sinh(\beta_{\min}/2)} + \frac{\beta_{\min}}{4\sinh^2(\beta_{\min}/4) }$ and $\beta_{\min} =\min_j \beta_j$. 

\end{theorem}

This theorem in the case of $m=1$ and for a non-optimal constant $\alpha$ is proven in~\cite{CarboneSasso2008}. The optimal constant $\alpha$ for $m=1$ is computed in~\cite{Beigi-RahimiKeshari2023}.

\begin{lemma}\label{lem:LSI-modified-Fisher}
Suppose that $\rho$ is a centered quantum state that is in Williamson's form. Let $\rho_G=\tau_1\otimes \cdots\otimes \tau_m$ be the Gaussification of $\rho$ with $\tau_j=(1-e^{-\beta_j})e^{-\beta_j \bfa_j^\dagger \bfa_j}$ and assume that $0<\beta_j<+\infty$ for all $j$.

Then,  there exists a constant $C>0$ such that
$$C\sum_{j=1}^m \tr\Big( \Big| e^{\beta_j/4} \ac_j \sqrt \rho - e^{-\beta_j/4} \sqrt \rho \ac_j   \Big|^2  \Big)\leq  J(\rho).$$
\end{lemma}

\begin{proof}
We present the proof only in the case of $m=1$ as the proof for arbitrary $m$ is similar. We use similar ideas as in the proof of Lemma~\ref{lem:Fisher-inf-bounds}.  Let
$$g(x, y)={\big{(} e^{\beta/4} \sqrt{y} - e^{-\beta/4} \sqrt{x}\big{)}}^{2}.$$
Then, by~\eqref{eq:def_Pi} we have
\begin{align*}\tr\Big( \Big| e^{\beta/4} \ac \sqrt \rho - e^{-\beta/4} \sqrt \rho \ac   \Big|^2  \Big) = \tr\big( \ac^{\dagger} {\PiInner{\rho}{g}} (\ac)\big).  
\end{align*}
We also have
\begin{align*}
\Big\| S_{\rho} + \frac{1}{\mu} \ac\Big\|_{2, \rho}^{2}  
	&= \tr\Big( \Big|\PiInner{\rho}{\sqrt{\psi}}\big{(} S_{\rho} + \frac{1}{\mu} \ac  \big{)}\Big|^{2} \Big)  \\
	&= \tr\Big( \Big| \PiInner{\rho}{\sqrt{\phi}} \big{(} [\ac, \rho ]\big{)} + \frac{1}{\mu} \PiInner{\rho}{\sqrt{\psi}}\big{(} \ac\big{)} \Big|^{2} \Big)  \\
	&= \tr\big( \big|  \PiInner{\rho}{\sqrt h}(\ac)  \big|^2 \big)\\
	&= \tr\big( \ac^{\dagger} {\PiInner{\rho}{h}} (\ac)\big)  ,
\end{align*}
where $h(x, y) = \big(\sqrt{\phi(x, y)} (y-x) + \frac{1}{\mu}\sqrt {\psi(x, y)}\big)^2$. Then, by Proposition~\ref{prop:pi-representation} we need to show that there is a constant $C>0$ such that 
$$ C g(x, y)\leq  h(x, y),$$
holds for all $x, y\geq 0$. It is shown in Appendix~\ref{ConstantFisherSobolev} that this inequality holds for 
$$C= \frac{8e^{-3\beta/2}}{(1+e^{-\beta})^2}.$$

\end{proof}

\subsection{Proof of Theorem~\ref{MainTheorem}}

First, by the same argument as in the proof of Theorem~\ref{TraceMainTheoremM} we may assume that $\rho$ is in Williamson's form, and $\nu_{\min}(\rho)>1$. 
The point is that as shown in Appendix~\ref{app:PoincareInequality}, $\lambda_\rho>0$ implies $\lambda_{U\rho U^\dagger}>0$ for any Gaussian unitary $U$. Then, we can assume that the covariance matrix of $\rho$ is diagonal. 
In this case, $\rho_G$ is a tensor product of thermal states of the form $\rho_G =  \tau_1\otimes \cdots \otimes \tau_m$ where $\tau_j = (1-e^{-\beta_j}) e^{-\beta_j \ac_j^\dagger \ac_j}$ for some $0<\beta_j<+\infty$. We note that $\beta_j<+\infty$ sine $\nu_j>1$. 
 
Then, by Theorem~\ref{thm:LSI} and Lemma~\ref{lem:LSI-modified-Fisher} there exists a constant $\alpha>0$ such that
$$\alpha D\big(\rho^{\boxplus n}\big\| \rho_G \big) \leq J\big(\rho^{\boxplus n}\big).$$ 
Thus, the desired result follows from Theorem~\ref{pseudoMainTheorem}.

%************************************************************************
\section{Successive projection lemma} \label{SuccessiveProjection}

In this section we develop some tools useful for the proof of Theorem~\ref{pseudoMainTheorem} which will be presented in the next section. Here, we prove a lemma which following~\cite{JB} we call the \emph{successive projection lemma}. To prove Theorem~\ref{pseudoMainTheorem} we will apply this lemma for the SLD score operator. 

For a function $g(x)$ and i.i.d random variables $X_1, \dots, X_n$, we may lift $g(x)$ to be a function on the tuple of variables $(X_1, \dots, X_n)$ simply by defining 
\begin{align}\label{eq:lifting-g-classic}
\widetilde g(X_1, \dots, X_n) := g\Big(\frac{X_1+\cdots + X_n}{\sqrt n}\Big).
\end{align}
One of the crucial steps in~\cite{JB} is the construction of successive projections of this function. Suppose that we want to bound the distance of the function $\widetilde g(X_1, \dots X_n)$ from its projection on the space of functions of the form $\frac{h(X_1) + \cdots + h(X_n)}{\sqrt{n}}$ that are summations of functions of individual variables. To this end, the optimal choice for $h(\cdot)$, when the distance is measured in terms of $2$-norm, is the projection of $\widetilde g(X_1, \dots X_n)$ on the space of functions of one of its variables.  
The successive projection lemma states that to estimate the distance of $\widetilde g(X_1, \dots X_n)$ to $\frac{h(X_1) + \cdots + h(X_n)}{\sqrt{n}}$ it is advantageous to consider projections of $\widetilde g(X_1, \dots X_n)$  not only on the space of single-variable functions, but also on the space of functions of any number of variables.  The idea is to consider the sequence of functions
\[
	 h_k  ( X_1 , \dots , X_k ) = \mathbb{E}_{X_{{k+1}}} \big{[}  h_{k+1} ( X_{1} , \dots , X_{k+1} ) \big{]},
\]
with $ h_n(\cdot) = \widetilde g(\cdot)$, in which case $h(\cdot)=h_1(\cdot)$. Then, the successive projection lemma of~\cite{JB} states that we may estimate the squared-distance 
\[
	{\mathbb{E} \bigg{[}  \Big(\widetilde g ( X_1 , \dots , X_n) - \frac {1}{\sqrt{n}}\sum_{\ell=1}^{n}h(X_\ell)\Big)^2 \bigg{]}},
\]
in terms of the squared-distances
\[
	{\mathbb{E} \bigg{[}  \Big( h_k  ( X_1 , \dots , X_k) - h_{k-1} ( X_1, \cdots, X_{k-1}) - \frac{1}{\sqrt{n}} h(X_k)  \Big)^{2}\bigg{]} }.
\]

Our goal in this section is to generalize the above ideas to the quantum case. To this end, we first need to generalize the lifting map~\eqref{eq:lifting-g-classic} to the quantum case.

\subsection{Symmetric lifting map} \label{TildeMap}

For any integer $n$ consider the set of unitaries 
$$\Big\{ e^{i\theta} \underbrace{D_{z/\sqrt{n}} \otimes \cdots \otimes D_{z/\sqrt{n}}}_{n \text{ times}}:\, z\in \mathbb C^m, \theta\in \mathbb R\Big\},$$
acting on the Hilbert space $\mathcal H_m^{\otimes n}$.
Observe that this set forms a group and the map $e^{i\theta}D_z\mapsto e^{i\theta}D_{z/\sqrt n}\otimes \cdots \otimes Z_{z/\sqrt n}$ is a representation of the Weyl-Heisenberg group $\{e^{i\theta} D_z:\, z\in \mathbb C^m, \theta\in \mathbb R\}$. Then, by the Stone--von Neumann theorem~\cite{Hall-Book} the Hilbert space $\mathcal H_m^{\otimes n}$ can be decomposed to an orthogonal direct sum of closed subspaces $\mathcal H_m^{\otimes n} =\bigoplus_r \mathcal W_r$ such that each subspace $\mathcal W_r$ is invariant under $D_{z/\sqrt n}^{\otimes n}$ for any $z\in \mathbb C^m$ and there is a unitary $U_r: \mathcal H_m\to \mathcal W_r$ such that
\begin{align}\label{eq:V-SvN}
D_{z/\sqrt{n}}^{\otimes n}\Big|_{W_r} = U_r D_z U_r^\dagger.
\end{align}
This means that $D_{z/\sqrt{n}}^{\otimes n}$ can be written as
$$D_{z/\sqrt{n}}^{\otimes n} = \bigoplus_r U_r D_z U_r^\dagger.$$
Extending this definition, for any (probably unbounded) operator $T$ acting on $\mathcal H_m$ we define
$$\widetilde T = \bigoplus_r U_r T U_r^\dagger.$$
Equivalently, we let $\widetilde T$ to be an operator acting on $\mathcal H_m^{\otimes n} =\bigoplus_r \mathcal W_r$ whose restriction on $\mathcal W_r$ equals $U_r TU_r^\dagger$. We call $T\mapsto \widetilde T$ the \emph{symmetric lifting map}.

Recall that a trace class operator can be represented in terms of its characteristic function as
\[
	T = \frac{1}{{\pi}^m} \int_{\mathbb{C}^{m}} \chi_{T}(z) D_{-z} \dd^{2m} z.
\]
Then, by the above discussion we have 
\begin{equation} \label{TildeMapEq}
	\widetilde{T} \coloneqq \frac{1}{{ \pi}^m} \int \chi_{T}(z) \underbrace{D_{-z/\sqrt{n}} \otimes \cdots \otimes D_{-z/\sqrt{n}}}_{n \text{ times}} \dd^{2m} z.
\end{equation}
As implied from this expression, $\widetilde{T}$ acts symmetrically on different subsystems, thus the name.

\begin{proposition} \label{Basictilde}
Let $\sigma=\rho_1\boxplus \cdots \boxplus \rho_n$ where $\rho_1, \dots,  \rho_n$ are $m$-mode quantum states with finite second moments. Then, the followings hold:
\begin{enumerate} [label={\rm (\roman*)}]
\item $\widetilde{{T}^{\dagger}} = \big(\widetilde{T}\big)^{\dagger}$ and $\widetilde{T} \widetilde{R} = \widetilde{TR}$.
\item $\tr\big(\rho_1\otimes \cdots \otimes \rho_n \widetilde{T}\big) = \tr\big(\sigma T\big)$ and $\langle  R, T\rangle_{\sigma}=\langle \widetilde R,  \widetilde T\rangle_{\rho_1\otimes \cdots \otimes \rho_n}$ for all operators $R, T$ satisfying $\|R\|_{2, \sigma}, \|T\|_{2, \sigma}<+\infty$.
\item We have $\widetilde \ac_j = \frac{\ac_{j,1}  + \cdots  +\ac_{j,n}}{\sqrt n}$, and  $\big\langle \widetilde T, \frac{\ac_{j,1}  + \cdots  +\ac_{j,n}}{\sqrt n}\big\rangle_{\rho_1\otimes \cdots \otimes \rho_n} = \langle T, \ac_j\rangle_{\sigma}$ if $\|T\|_{2, \sigma}<+\infty$. 
\item $\widetilde{[\ac_j, T]} = \sqrt n [\ac_{j, \ell}, \widetilde T]$ for every $1\leq \ell\leq n$ if $\|T\|_{2, \sigma}<+\infty$.
\end{enumerate}
\end{proposition}

\begin{proof}
(i) is an immediate consequence of the definition. 
To prove (ii) we first assume that $T$ is trace class and use $\chi_{\sigma}(z) = \prod_{\ell=1}^n \chi_{\rho_\ell}(z/\sqrt n)^n$ to compute
\begin{align*}
    \tr\big(\rho_1\otimes \cdots\otimes \rho_n \widetilde{T}\big)
    &= \frac{1}{{(2 \pi)}^m} \int \chi_{T}(z) \prod_{\ell=1}^n \tr\big(\rho_\ell D_{-z/\sqrt{n}}\big)^n \dd^{2m} z  \\
    &= \frac{1}{{(2 \pi)}^m} \int \chi_{T}(z) \prod_{\ell=1}^n\chi_{\rho_\ell}(-z/\sqrt{n})^n  \dd^{2m} z \\
    &= \frac{1}{{(2 \pi)}^m} \int \chi_{T}(z) \chi_{\sigma}(-z)  \dd^{2m} z \\
    &= \frac{1}{{(2 \pi)}^m} \int \chi_{T}(z) \tr\big(\sigma  D_{-z}\big) \dd^{2m} z \\
    &= \tr\big(\sigma T\big).
  \end{align*}
Then, $\langle R, T\rangle_{\sigma} = \langle \widetilde R, \widetilde T\rangle_{\rho_1\otimes \cdots \otimes \rho_n}$ follows from (i) assuming that both $R, T$ are trace class. The general case is proven by continuity.  For (iii) start with $\widetilde{D_z} = D_{z/\sqrt n}^{\otimes n}$ and compute the derivative $\partial_{\bar z}$ of both sides at $z=0$. Then, the second part follows from (ii). Finally, for (iv) by comparing $\big[\ac_{j, \ell}, D_{\frac {z}{\sqrt n}}^{\otimes n}\big] = \frac {z_j}{\sqrt n} D_{\frac {z}{\sqrt n}}^{\otimes n}$ and $[\ac_j, D_{z}] = z_j  D_{z}$ we find that the equality holds if $T$ is trace class. The more general case then follows by a continuity argument.
\end{proof}

The function $\widetilde g(x_1, \dots x_n)=g(\frac{x_1 + \cdots + x_n}{\sqrt{n}})$ is symmetric under the permutation of its coordinates. By the definition of $\widetilde T$, as an operator acting on $\mathcal H_m^{\otimes n}$, it is invariant under the permutation of subsystems. Moreover, the derivative  $\partial_{x_\ell}\widetilde g(x_1, \dots, x_n)$ is independent of $1\leq \ell\leq n$. 
Part (iv) of Proposition~\ref{Basictilde} is a quantum generalization of this observation.

%******************************************************************
\subsection{Projections onto additive subspaces} 
Having the powerful tool of the symmetric lifting map we are now ready to establish the successive projection lemma in the quantum case. To this end, let $\rho$ be an $m$-mode quantum state that is centered, meaning that $\tr(\rho \ac_j)=0$. 
For simplicity of notation we let
\[
	\sigma := \rho^{\boxplus n}.
\]
Next, 
for an operator $X$ satisfying $\|X\|_{2, \sigma}<+\infty$ we consider $\widetilde{X}$ as defined in the previous subsection. Let
\begin{align}\label{eq:def-Y-ell}
	Y_{\ell} := \sqrt{n} \, \tr_{\neg \ell}\big(\rho_1 \otimes \dots \otimes \rho_{\ell-1} \otimes \rho_{\ell+1} \otimes \dots \otimes \rho_{n} \widetilde{X}\big).
\end{align}
Here, $\rho_\ell$ is the same as $\rho$, but acting on the $\ell$-th subsystem. Also, $\tr_{\neg \ell}(\cdot)$ denotes the partial trace with respect to all but the $\ell$-th subsystem. The following proposition shows that this partial trace is indeed well-defined. 

\begin{proposition} \label{BasicY}
$Y_{\ell}$ is the orthogonal projection of $\sqrt{n} \widetilde{X}$ on the $\ell$-th subsystem. In other words, for any operator $V_\ell$ acting on the $\ell$-th subsystem, we have
\[
	{\langle V_\ell , Y_{\ell}  \rangle}_{\rho_{\ell}} = \sqrt{n} {\langle V_\ell , \widetilde{X} \rangle}_{\rho^{\otimes n}}.
\]
\end{proposition}

\begin{proof}
Straightforward computations yield
\begin{align*}
	 \langle V_\ell , \widetilde{X} \rangle_{\rho^{\otimes n}} 
	 &= \frac{1}{2}\tr\big(\rho_1 \otimes \dots \otimes \rho_{n} {V_\ell}^{\dagger} \widetilde{X}\big) + \frac 12\tr\big( {V_\ell}^{\dagger} \rho_1 \otimes \dots \otimes \rho_{n} \widetilde{X}\big)  \\
	 &= \frac{1}{2}  \tr\Big(\rho_\ell {V_\ell}^{\dagger} \tr_{\neg \ell}\big(\rho_1 \otimes \dots \otimes \rho_{\ell-1} \otimes \rho_{\ell+1} \dots \otimes \rho_{n} \widetilde{X}\big)\Big) \\
	 &\quad\,  + \frac12 \tr\Big({V_\ell}^{\dagger} \rho_\ell \tr_{\neg \ell}\big(\rho_1 \otimes \dots \otimes \rho_{\ell-1} \otimes \rho_{\ell+1} \dots \otimes \rho_{n} \widetilde{X}\big)\Big) \\
	 &= \frac{1}{2\sqrt n} \tr\big(\rho_\ell {V_\ell}^{\dagger} Y_{\ell}\big) + \frac {1}{2\sqrt n}\tr\big( {V_\ell}^{\dagger} \rho_\ell Y_{\ell}\big) \\
	 &= \frac{1}{\sqrt{n}}\langle V_\ell , Y_{\ell} \rangle_{\rho_{\ell}}.
\end{align*}
\end{proof}

We also define
\begin{align}\label{eq:def-Z-ell}
	Z_{[\ell]} := \tr_{\ell+1, \dots, n}\big(\rho_{\ell+1} \otimes \dots \otimes \rho_n \widetilde{X}\big).
\end{align}
By a similar computation as above, it can be verified that $Z_{[\ell]}$ is the projection of $\widetilde{X}$ on the first $\ell$ subsystems. After defining these projections onto various subsystems, we can present the generalization of~\cite[Lemma 3.3]{JB} in the quantum setting.

\begin{lemma} \label{LowerRtildeLemma} \emph{[Successive projection lemma]}
Let $\rho$ be a centered quantum state that is in Williamson's form. Assume that $\rho$ satisfies the quantum Poincar\'e inequality with constant $\lambda_\rho>0$.
Let $X$ be an operator that satisfies $\tr(\sigma X) =0$ and $\|X\|_{2,\sigma}<+\infty$.   
Then, for any $1\leq \ell\leq n$ we have 
	\[
	\frac{\lambda_{\rho} (\ell-1)}{ 2n I(\rho) } \Big\| Y_{\ell} + \sum_{j=1}^m\Big( {\langle S_{\sigma,j}, X \rangle}_{\sigma} \ac_{j,\ell}-  {\langle X^\dagger, S_{\sigma,j} \rangle}_{\sigma} \ac^\dagger_{j,\ell} \Big) \Big\|_{2, \rho_{\ell}}^{2}\leq 	\Big\| Z_{[\ell]} - Z_{[\ell-1]} - \frac{1}{\sqrt{n}} Y_{\ell}\Big\|_{2, \rho_1 \otimes \dots \otimes \rho_{\ell}}^{2}.
	\]
Here, $\ac_{j,\ell}$ is the annihilator operator of the $j$-th mode of the $\ell$-th subsystem, $\sigma= \rho^{\boxplus n}$  and $Y_\ell, Z_{[\ell]}$ are defined in~\eqref{eq:def-Y-ell} and~\eqref{eq:def-Z-ell},  respectively.
\end{lemma}

\begin{proof}
For simplicity we first present the proof in the case of $m=1$.  This saves us from repeating index $j$, so only need to be worried about indices $\ell, k$ representing subsystems. Here, we use $S_{\rho, k} = \PiInner{\rho_{k}}{\phi}([\ac_{k}, \rho_{k}])$ for the SLD score operator of  $\rho$ acting on the $k$-th subsystem, and $S_{\sigma} = \PiInner{\sigma}{\phi}([\ac, \sigma])$ is the SLD score operator of $\sigma$.

Let $T_{\ell}$ be the operator acting on $\ell$-th subsystem given by
\begin{align*}
		T_{\ell} := &\, \tr_{1, \dots , \ell-1}\bigg(\PiInner{\rho_{1} \otimes \cdots \otimes \rho_{\ell-1}}{\psi}\big(S_{\rho, 1}^{\dagger} + \cdots + S_{\rho, \ell-1}^{\dagger}\big) \Big(Z_{[\ell]} - Z_{[\ell-1]} - \frac{1}{\sqrt{n}} Y_{\ell}\Big)\bigg)\\
		 =&\, \tr_{1, \dots , \ell-1}\Big(\PiInner{\rho_{1} \otimes \cdots \otimes \rho_{\ell-1}}{\psi}\big(S_{\rho, 1}^{\dagger} + \cdots + S_{\rho, \ell-1}^{\dagger}\big) Z_{[\ell]}\Big)\\
		&\, - \tr_{1, \dots , \ell-1}\Big(\PiInner{\rho_{1} \otimes \cdots \otimes \rho_{\ell-1}}{\psi}\big(S_{\rho, 1}^{\dagger} + \cdots + S_{\rho, \ell-1}^{\dagger}\big) Z_{[\ell-1]}\Big)\\
		&  \, - \frac{1}{\sqrt n} \tr_{1, \dots , \ell-1}\Big(\PiInner{\rho_{1} \otimes \cdots \otimes \rho_{\ell-1}}{\psi}\big(S_{\rho, 1}^{\dagger} + \cdots + S_{\rho, \ell-1}^{\dagger}\big) Y_\ell\Big),
\end{align*}
where as before $\psi(x,y) = \frac{x+y}{2}$.
	We begin with the first term in order to obtain an equivalent expression for $T_\ell$:
	\begin{align*}
	\tr_{1, \dots , \ell-1}\Big(&\PiInner{\rho_{1} \otimes \cdots \otimes \rho_{\ell-1}}{\psi}(S_{\rho, 1}^{\dagger}) Z_{[\ell]} \Big) \\
	&= \tr_{1, \dots , \ell-1, \ell+1, \dots, n}\Big(\PiInner{\rho_{1} \otimes \cdots \otimes \rho_{\ell-1} \otimes \rho_{\ell+1} \otimes \dots \otimes \rho_{n}}{\psi}(S_{\rho, 1}^{\dagger}) \widetilde{X}\Big) \\
		&=   \tr_{1, \dots , \ell-1, \ell+1, \dots, n}\Big(\PiInner{\rho_{1}}{\psi}  (S_{\rho, 1}^{\dagger}) \otimes \rho_{2} \otimes \cdots \otimes \rho_{\ell-1} \otimes \rho_{\ell+1} \otimes \dots \otimes \rho_{n}  \widetilde{X}\Big)  \\
	&=   \tr_{1, \dots , \ell-1, \ell+1, \dots, n}\Big([\ac_{1}, \rho_1]^{\dagger} \otimes \rho_{2} \otimes \cdots \otimes \rho_{\ell-1} \otimes \rho_{\ell+1} \otimes \dots \otimes \rho_{n} \widetilde{X}\Big)  \\
	&= \tr_{1, \dots , \ell-1, \ell+1, \dots, n}\Big(\rho_{1} \otimes \cdots \otimes \rho_{\ell-1} \otimes \rho_{\ell+1} \otimes \dots \otimes \rho_{n} [{\ac_{1}^{\dagger}} , \widetilde{X}]\Big)  \\
	&=  \tr_{1, \dots , \ell-1, \ell+1, \dots, n}\Big(\rho_{1} \otimes \cdots \otimes \rho_{\ell-1} \otimes \rho_{\ell+1} \otimes \dots \otimes \rho_{n} [{\ac_{\ell}^{\dagger}} , \widetilde{X}]\Big)  \\
	&= \frac{1}{\sqrt{n}}\big[{\ac_{\ell}^{\dagger}}, Y_{\ell}\big].
	\end{align*}
Here, in the third line we use the fact that $S_{\rho, 1}^{\dagger}$ acts on only on the first subsystem and commutes with all $\rho_k$ with $k\neq 1$. Also, the penultimate line follows from Proposition~\ref{Basictilde}. The same computation works if we replace  $S_{\rho, 1}^{\dagger}$  with $S_{\rho, k}^{\dagger}$ for  any $1\leq k\leq \ell-1$. Then, we have
	\begin{align*}
	\tr_{1, \dots , \ell-1}\Big(\PiInner{\rho_{1} \otimes \cdots \otimes \rho_{\ell-1}}{\psi}(S_{\rho, 1}^{\dagger} + \cdots + S_{\rho, \ell-1}^{\dagger}) Z_{[\ell]}\Big) 	= \frac{\ell-1}{\sqrt{n}} [{\ac_{\ell}^{\dagger}}, Y_{\ell}].
	\end{align*}
For the second term we have
	\begin{align*}
	 \tr_{1, \dots , \ell-1}\Big(&\PiInner{\rho_{1} \otimes \cdots \otimes \rho_{\ell-1}}{\psi}(S_{\rho, 1}^{\dagger} + \cdots + S_{\rho, \ell-1}^{\dagger}) Z_{[\ell-1]}\Big) \\
	 	&= \sum_{k=1}^{\ell-1} \tr\Big(\PiInner{\rho_{1} \otimes \cdots \otimes \rho_{\ell-1}}{\psi}({S_{\rho, k}}^{\dagger}) \otimes \rho_{\ell}\cdots\otimes  \rho_n \widetilde{X}\Big) \\
	&= \sum_{k=1}^{\ell-1} \tr\Big(\rho_{1} \otimes \dots \otimes \rho_{k-1} \otimes [{\ac_{k}}, \rho_k]^{\dagger} \otimes \rho_{k+1} \otimes \dots \otimes \rho_{n} \widetilde{X}\Big) \\
		&= \sum_{k=1}^{\ell-1} \tr\Big(\rho_{1} \boxplus \dots \boxplus \rho_{k-1} \boxplus [{\ac_{k}}, \rho_k]^{\dagger} \boxplus \rho_{k+1} \boxplus \dots \boxplus \rho_{n} X\Big) \\
				&= \sum_{k=1}^{\ell-1} \frac{1}{\sqrt n}\tr\Big( [\ac , \sigma]^\dagger X\Big) \\
				&= \frac{\ell-1}{\sqrt{n}} {\langle S_{\sigma}, X \rangle}_{\sigma}.
	\end{align*}
Here, the fourth line follows from Proposition~\ref{Basictilde} and the fifth line follows from~\eqref{DeriveConvolution}. 
Finally, for the third term we have
	\begin{align*}
\frac{1}{\sqrt n} \tr_{1, \dots , \ell-1} \Big(\PiInner{\rho_{1} \otimes \cdots \otimes \rho_{\ell-1}}{\psi} & \big(S_{\rho, 1}^{\dagger} + \cdots + S_{\rho, \ell-1}^{\dagger}\big) Y_\ell\Big)\\
	&= \frac{1}{\sqrt{n}} \sum_{k=1}^{\ell-1}  \tr\Big(\PiInner{\rho_{1} \otimes \cdots \otimes \rho_{\ell-1}}{\psi}(S_{\rho, k}^{\dagger} )\Big)  \, Y_{\ell}\\
		&= \frac{1}{\sqrt{n}} \sum_{k=1}^{\ell-1}  \tr\big( [\ac_k, \rho_k]^\dagger\big)  \, Y_{\ell}\\
	&= 0.
	\end{align*}
Putting these together and using the canonical commutation relations~\eqref{CCR},  we find that
	\begin{equation}\label{refNoteT}
		T_{\ell} =\frac{\ell-1}{\sqrt{n}}\big[{\ac_{\ell}^{\dagger}},  Y_{\ell} + {\langle S_{\sigma}, X \rangle}_{\sigma} \ac_{\ell} - {\langle X^\dagger, S_{\sigma} \rangle}_{\sigma} \ac_{\ell}^\dagger \big].
	\end{equation}
	
Next, define	
\begin{align*}
		T'_{\ell} :=& \tr_{1, \dots , \ell-1}\bigg(\PiInner{\rho_{1} \otimes \cdots \otimes \rho_{\ell-1}}{\psi}\big(-S_{\rho, 1} - \cdots - S_{\rho, \ell-1}\big) \Big(Z_{[\ell]} - Z_{[\ell-1]} - \frac{1}{\sqrt{n}} Y_{\ell}\Big)\bigg).
\end{align*}	
Following the same computations as above we also find that
$$		T'_{\ell} =\frac{\ell-1}{\sqrt{n}}\big[\ac_{\ell},  Y_{\ell} + {\langle S_{\sigma}, X \rangle}_{\sigma} \ac_{\ell} - {\langle X^\dagger, S_{\sigma} \rangle}_{\sigma} \ac_{\ell}^\dagger \big].
$$	
Recall that $\tr(\rho_\ell  Y_\ell) = \tr(\rho X)=0$. Also, since $\rho$ is centered, we have $\tr(\rho\ac_\ell)=\tr(\rho \ac_\ell^\dagger)=0$. 
Therefore, we may apply the Poincar\'e inequality on operator $Y_{\ell} + {\langle S_{\sigma}, X \rangle}_{\sigma} \ac_{\ell} - {\langle X^\dagger, S_{\sigma} \rangle}_{\sigma} \ac_{\ell}^\dagger$ to conclude that
	\begin{align} \label{PoinTT}
		{\| T_{\ell} \|}_{2, \rho_{\ell}}^{2} + {\| T'_{\ell} \|}_{2, \rho_{\ell}}^{2}			\geq \lambda_{\rho}   \frac{(\ell-1)^2}{n}  \big\| Y_{\ell} + {\langle S_{\sigma}, X \rangle}_{\sigma} \ac_{\ell}- {\langle X^\dagger, S_{\sigma} \rangle}_{\sigma} \ac_{\ell}^\dagger \big\|_{2, \rho_{\ell}}^{2}.
	\end{align}
To bound the left hand side, by the extension of the Cauchy-Schwarz inequality proven in Lemma~\ref{lem:Cauchy-Scwarz} of Appendix~\ref{CauchyShwarzT} we have \begin{align*}
\max \big\{ \|T_{\ell}\|_{2, \rho_\ell}^2, \|T'_{\ell}\|_{2, \rho_\ell}^2	 \big\}
		& \leq  \big\| S_{\rho, 1} + \cdots + S_{\rho, \ell-1} \big\|^2_{2, \rho_1\otimes \cdots\otimes \rho_{\ell-1}}   \Big\| Z_{[\ell]} - Z_{[\ell-1]} - \frac{1}{\sqrt{n}} Y_{\ell}\Big\|^2_{2, \rho_1\otimes \cdots \otimes \rho_\ell}\nonumber\\
		& = (\ell-1)\big\| S_{\rho}  \big\|_{2, \rho}^2\cdot    \Big\| Z_{[\ell]} - Z_{[\ell-1]} - \frac{1}{\sqrt{n}} Y_{\ell}\Big\|_{2, \rho_1\otimes \cdots \otimes \rho_\ell}^2,
\end{align*}	
where for the second line we use 
$$\langle S_{\rho, k}, S_{\rho, k'}\rangle_{\rho_1\otimes \cdots\otimes \rho_{\ell-1}} = \tr\big(\PiInner{\rho_{k}}{\psi} S_{\rho, k}^\dagger \big)\cdot \tr\big(\PiInner{\rho_{k'}}{\psi} S_{\rho, k'} \big) = \tr\big( [\ac_k, \rho_k]^\dagger \big)\cdot \tr\big( [\ac_{k'}, \rho_{k'}] \big)=0,$$
if $ k\neq k'$.
Comparing these to~\eqref{PoinTT} we arrive at
\begin{align*}
	   \frac{\lambda_{\rho}(\ell-1)}{2n \|S_{\rho}\|_{2, \rho}^2} \cdot \big\| Y_{\ell} + {\langle S_{\sigma}, X \rangle}_{\sigma} \ac_{\ell}  - {\langle X^\dagger, S_{\sigma} \rangle}_{\sigma} \ac_{\ell}^\dagger\big\|_{2, \rho_{\ell}}^{2}\leq     \Big\| Z_{[\ell]} - Z_{[\ell-1]} - \frac{1}{\sqrt{n}} Y_{\ell}\Big\|_{2, \rho_1\otimes \cdots \otimes \rho_\ell}^2,
\end{align*}
as desired.

To prove the theorem in the multimode case we need to define $T_{j, \ell},  T'_{j, \ell}$ for any mode $1\leq j\leq m$. Then, similarly as in the single mode case we have
\begin{align*}
T_{j,\ell} &=\frac{\ell-1}{\sqrt{n}}\big[{\ac_{j,\ell}^{\dagger}},  Y_{\ell} + {\langle S_{\sigma, j}, X \rangle}_{\sigma} \ac_{j,\ell} - {\langle X^\dagger, S_{\sigma} \rangle}_{\sigma} \ac_{\ell}^\dagger \big]\\
&=\frac{\ell-1}{\sqrt{n}}\Big[{\ac_{j,\ell}^{\dagger}},  Y_{\ell} + \sum_{j'=1}^m\Big( {\langle S_{\sigma, j'}, X \rangle}_{\sigma} \ac_{j',\ell} - {\langle X^\dagger, S_{\sigma, j'} \rangle}_{\sigma} \ac_{j', \ell}^\dagger \Big)\Big],
\end{align*}
where the second line follows from canonical commutation relations. 
A similar expression holds for $T'_{j, \ell}$. Then, the proof goes through by applying the Poincar\'e inequality as before.

\end{proof}

%***************************************************************
\section{Convergence rate for SLD Fisher information distance}  \label{ConvergenceFID}

In this section we give the proof of Theorem~\ref{pseudoMainTheorem}. Here, is a more formal statement of the theorem. 

\begin{theorem} \label{pseudoMainTheorem-v2}
	Let $\rho$ be a centered $m$-mode quantum state that satisfies the quantum Poincar\'e inequality with constant $\lambda_{\rho}>0$. Then, we have
\begin{align}\label{eq:Fisher-inf-bound-exact}
		J\big(\rho^{\boxplus n}\big)\leq 	\frac{1}{1 +   \frac{\lambda_\rho}{4 I(\rho)} (n-1) } J(\rho).
	\end{align}
\end{theorem}

The rest of this section is devoted to the proof of this theorem.
Similarly to the previous section, we first present the proof in the single-mode case. Later, we will explain modifications needed in the multimode case.  We also assume with no loss of generality that $\rho$ is in Williamson's form. 

It is helpful to keep a geometric picture in mind. Observe that by~\eqref{eq:gamma-d-conv} we have $\|\ac\|_{2, \sigma}^2 = \|\ac\|_{2, \rho}^2 = \mu$. Next, $\widetilde S_\sigma$ is an operator on the tensor product space, which by  Proposition~\ref{Basictilde} satisfies
\begin{align}\label{eq:J-sigma-0}
J(\sigma) = I(\sigma) - \frac 1\mu =\Big\|S_{\sigma} + \frac{1}{\mu} \ac\Big\|_{2, \sigma}^2  = \bigg\|\widetilde S_{\sigma} + \frac{1}{\mu} \frac{\ac_1+\cdots +\ac_n}{\sqrt n}\bigg\|_{2, \rho^{\otimes n}}^2.
\end{align}
Our goal is to show that $J(\sigma)$ is small. That is, we want to prove that $\widetilde S_{\sigma}$ is close to $-\frac{1}{\mu}\frac{\ac_1+\cdots+\ac_n}{\sqrt n}$. On the other hand, it is easily verified that 
\begin{align}\label{eq:J-rho-0}
J(\rho)= I(\rho) - \frac 1\mu = \Big\|S_{\rho} + \frac{1}{\mu} \ac\Big\|_{2, \rho}^2 =  \bigg\|\frac{S_{\rho, 1} + \cdots +S_{\rho, n}}{\sqrt n} + \frac{1}{\mu} \frac{\ac_{\rho, 1} + \cdots +\ac_{\rho, n}}{\sqrt n} \bigg\|_{2, \rho^\otimes n}^2.
\end{align}
The three points 
$$\widetilde S_{\sigma},\quad  \widehat S_{\rho}:= \frac{S_{\rho, 1} + \cdots +S_{\rho, n}}{\sqrt n}, \qquad  -\frac{1}{\mu}\widehat \ac:=-\frac{1}{\mu} \frac{\ac_{\rho, 1} + \cdots +\ac_{\rho, n}}{\sqrt n},$$ 
form a triangle in the Euclidean space of operators acting on the tensor product space equipped with inner product $\langle \cdot, \cdot \rangle_{\rho^{\otimes n}}$. Moreover, the former two points belong to the subspace $\mathcal W$ of operators that are of the form $T_1 + \cdots + T_n$ where $T_\ell$ acts only on the $\ell$-th subsystem and are orthogonal to the identity operator:
\[
	\mathcal{W} := \big\{ T_1 + \cdots + T_n:\, T_\ell \text{ acts only on } \ell\text{-th subsystem and }  {\langle \mathbb I_\ell , T_\ell \rangle}_{\rho_\ell} = 0, \forall \ell\big\}.
\]
See Figure~\ref{figure}.

%%%%%%
\begin{figure}[t]
  \begin{minipage}{.5\textwidth}
    \centering
    \includegraphics[width=3.4in]{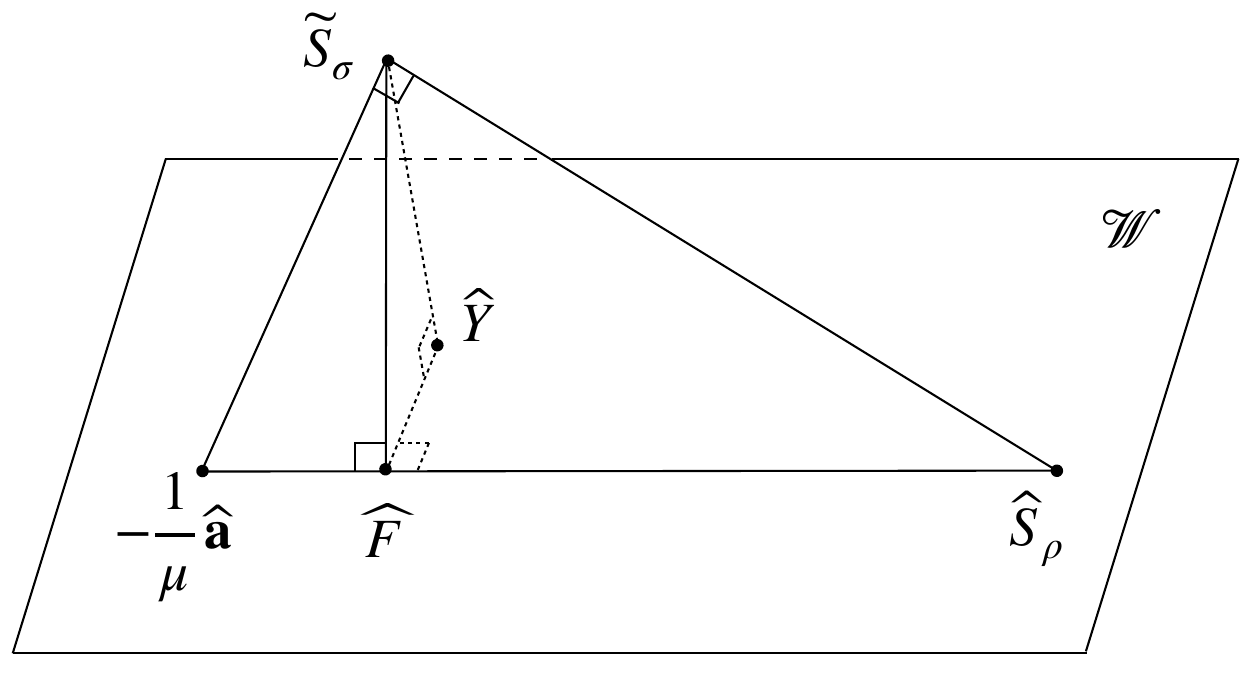}
  \end{minipage}
  \begin{minipage}{.49\textwidth}
   {\footnotesize 
    \begin{align*}
    	& \\
	&\\
        &\big\|\widetilde S_\sigma  + \frac{1}{\mu} \widehat \ac\big\|_{2, \rho^{\otimes n}}^2 = J(\sigma)\\
        &\big\|\widehat S_\rho  + \frac{1}{\mu} \widehat \ac\big\|_{2, \rho^{\otimes n}}^2 = J(\rho)\\
        &\big\|\widehat S_\rho  - \widehat S_\rho\big\|_{2, \rho^{\otimes n}}^2 = J(\rho)- J(\sigma)\\
        & \big\|\widehat F  + \frac{1}{\mu} \widehat \ac \big\|_{2, \rho^{\otimes n}}^2 = \frac{J(\sigma)^2}{J(\rho)}\\
        &\big\| \widetilde{S}_{\sigma} -  \widehat Y \big\|_{2, \rho^{\otimes n}}^2\leq J(\sigma)-\frac{J(\sigma)^2}{ J(\rho)}\\
        & \big\|  \widehat Y +\frac 1\mu \widehat \ac\big\|_{2, \rho^{\otimes n}}^2\geq  \frac{J(\sigma)^2}{J(\rho)}
    \end{align*}
    }
  \end{minipage}%
\caption{\small Geometric illustration of operators acting on the $n$-fold tensor product space. Here,
$S_{\sigma}$ is the SLD score operator of $\sigma=\rho^{\boxplus n}$ and $\widetilde S_\sigma$ its image under the symmetric lifting map. $\widehat S_\rho = \frac{1}{\sqrt n}\big( S_{\rho, 1} + \cdots + S_{\rho, n} \big)$ and  $\widehat \ac = \frac{1}{\sqrt n}\big( \ac_{ 1} + \cdots + \ac_{n} \big)$. The projection of $\widetilde S_\sigma$ on the line segment $(-\frac {1}{\mu}\widehat \ac, \widehat S_\rho)$ is denoted by  $\widehat F$ and takes the form $\widehat F = \frac{1}{\sqrt n} \big(F_1+\cdots + F_n\big)$. The projection of $\widetilde S_{\sigma}$ on  the subspace $\mathcal W$ consisting of operators of the form $T_1+\cdots + T_n$, where  $T_\ell$ acts only on the $\ell$-th subsystem and is orthogonal to the identity operator, is denoted by $\widehat Y = \frac{1}{\sqrt n}  \big(Y_1+\cdots+Y_n\big)$.  
} 
\label{figure}
\end{figure}
%%%%%%

By~\eqref{eq:J-sigma-0} and~\eqref{eq:J-rho-0} the squared length of two line segments of the triangle $\big(\widetilde S_\sigma, \widehat S_\rho, -\frac{1}{\mu} \widehat \ac\big)$ are $J(\rho), J(\sigma)$. It will be fruitful to calculate the squared length of the third line segment. To this end, using Proposition~\ref{Basictilde} we first compute
\begin{align*}  
	\Big\langle	 \widehat S_\rho , \widetilde{S}_{\sigma} \Big\rangle_{\rho^{\otimes n}} 
	&= \sqrt{n}\, \big\langle S_{\rho, 1} , \widetilde{S}_{\sigma} \big\rangle_{\rho^{\otimes n}}  \\
	&= \sqrt{n} \, \tr\big({[{\ac_{1}}, \rho_1]}^{\dagger} \otimes \rho_2 \cdots \otimes \rho_{n}  \widetilde{S}_{\sigma}\big)  \\
	&= \sqrt{n}\, \tr\big({[{\ac_{1}}, \rho_1]}^{\dagger} \boxplus \dots \boxplus \rho_{n} \, S_{\sigma}\big)  \\
	&= \tr\big( {[\ac, \sigma]}^{\dagger} S_{\sigma}\big) \\
	&= {\|S_{\sigma}\|}_{2, \sigma}^2 \\
	& = I(\sigma),
	\end{align*}
where the fourth line follows from~\eqref{DeriveConvolution}. Therefore,
\begin{align*}
\big\| \widetilde{S}_{\sigma} - \widehat S_\rho \big\|_{2, \rho^{\otimes n}}^{2} & = \big\| \widetilde{S}_{\sigma}\big\|_{2, \rho^{\otimes n}}^{2} + \big\| \widehat S_\rho \big\|_{2, \rho^{\otimes n}}^{2} - 2 I(\sigma) \\
& = I(\rho) - I(\sigma)\\
& = J(\rho) - J (\sigma).
\end{align*}
Comparing to~\eqref{eq:J-sigma-0} and~\eqref{eq:J-rho-0} and using the Pythagorean theorem we realize that, as depicted in Figure~\ref{figure}, the three points $\widetilde S_{\sigma}, \widehat  S_\rho$ and $-\frac{1}{\mu} \widehat  \ac$ form a \emph{right triangle}. 

Consider the closest point to $\widetilde S_\sigma$ on the line segment $\big(\widehat S_\rho, -\frac{1}{\mu} \widehat\ac \big)$. Recall that this line segment belongs to $\mathcal W$. Then, this closest point is of the form 
$$\widehat F:=\frac{F_{1} + \cdots +F_{n}}{\sqrt n},$$
where $F_\ell$ acts only on the $\ell$-th subsystem. Note that by symmetry, $F_{1}, \dots ,F_{n}$ are the same operators but acting on different subsystems.
Since $\big(\widetilde S_{\sigma}, \widehat S_\rho, -\frac{1}{\mu} \widehat \ac\big)$ is a right triangle, by elementary Euclidian geometry we have\footnote{This fact can also be proven by a simple optimization problem corresponding to the definition of $ \frac{F_1 + \cdots +F_n}{\sqrt n}$ as in~\cite{JB}.}
\[
	\Big\| \widehat  F+ \frac{1}{\mu}\widehat \ac \Big\|_{2, \rho^{\otimes n}}  =    \frac{\Big\|\widetilde S_{\sigma} + \frac{1}{\mu} \widehat \ac\Big\|_{2, \rho^{\otimes n}}^2  }{\Big\|\widehat S_\rho + \frac{1}{\mu} \widehat \ac \Big\|_{2, \rho^{\otimes n}}} =  \frac{J(\sigma)}{\sqrt {J(\rho)}}.
\]

Let
\[
	Y_{\ell} = \sqrt{n} \, \tr_{\neg \ell}\big(\rho_1 \otimes \dots \otimes \rho_{\ell-1} \otimes \rho_{\ell+1} \dots \otimes \rho_{n} \widetilde S_\sigma\big).
\]
Then, by Proposition~\ref{BasicY}, the operator
$$\widehat Y:=\frac{Y_1+\cdots + Y_n}{\sqrt n},$$ 
is the closest point to $\widetilde S_\sigma$ on the hyperplane $\mathcal W$, and the line segment $(\widetilde S_\sigma, \widehat  Y)$ is perpendicular to $\mathcal W$. Now, since $\widehat F$ also belongs to $\mathcal W$, by Pythagorean's theorem we have
\begin{align}\label{eq:J-S-F-Y-bound}
{\big\| \widetilde{S}_{\sigma} -  \widehat Y \big\|}_{2, \rho^{\otimes n}}^2 & \leq {\big\| \widetilde{S}_{\sigma} -  \widehat F \big\|}_{2, \rho^{\otimes n}}^2 \nonumber\\ 
& = \Big\| \widetilde S_\sigma + \frac{1}{\mu}\widehat \ac  \Big\|_{2, \rho^{\otimes n}}^2 - \Big\| \widehat F + \frac{1}{\mu}\widehat \ac  \Big\|_{2, \rho^{\otimes n}}^2\nonumber\\ 
&=J(\sigma)-\frac{J(\sigma)^2}{ J(\rho)}.
\end{align}
Therefore, applying the Pythagorean theorem for the right triangle $\big(\widetilde S_\sigma,  \widehat Y, -\frac 1\mu\widehat \ac\big)$ and using the above inequality we obtain 
\begin{equation} \label{RGUpperBound}
	\Big\|  \widehat Y +\frac 1\mu \widehat \ac\Big\|_{2, \rho^{\otimes n}}^2\geq  \frac{J(\sigma)^2}{J(\rho)}.
\end{equation}

Next, to lower bound the left hand side of~\eqref{eq:J-S-F-Y-bound}, we use Lemma~\ref{LowerRtildeLemma} for the choice of 
$$X = S_\sigma.$$
We note that $\tr(\sigma S_\sigma) = \tr([\ac, \sigma])=0$ and $\|S_\sigma\|_\sigma^2=I(\sigma)<+\infty$. Therefore,
\begin{align}\label{eq:apply-LowerRtildeLemma-S-sigma}
	   \frac{\lambda_{\rho}(\ell-1)}{2n I(\rho)} \cdot \big\| Y_{\ell} + I(\sigma) \ac_{\ell}  -\langle S_{\sigma}^\dagger, S_\sigma\rangle_\sigma \ac^\dagger_{\ell} \big\|_{2, \rho_{\ell}}^{2}\leq     \Big\| Z_{[\ell]} - Z_{[\ell-1]} - \frac{1}{\sqrt{n}} Y_{\ell}\Big\|_{2, \rho_1\otimes \cdots \otimes \rho_\ell}^2,
\end{align}
for every $1\leq \ell\leq n$, where
\[
	Z_{[\ell]} := \tr_{\ell+1, \dots, n}\big(\rho_{\ell+1} \otimes \dots \otimes \rho_n \widetilde S_\sigma\big).
\]
We need to calculate the right hand side of~\eqref{eq:apply-LowerRtildeLemma-S-sigma}. To  this end, we compute:
	\begin{align*}
		\big\langle Z_{[\ell]} , Z_{[\ell-1]} \big\rangle_{\rho_{1} \otimes \dots \otimes \rho_{\ell}} 
		& = \frac{1}{2} \tr\big(\rho_{1} \otimes \dots \otimes \rho_{\ell} Z_{[\ell]}^{\dagger} Z_{[\ell-1]}\big)  + \frac 12 \tr\big( Z_{[\ell]}^{\dagger} \rho_{1} \otimes \dots \otimes \rho_{\ell} Z_{[\ell-1]}\big) \\
		&= \frac{1}{2} \tr_{1, \dots , \ell-1}\big(\rho_{1} \otimes \cdots \otimes \rho_{\ell-1} \tr_{\ell}(\rho_{\ell} Z_{[\ell]}^{\dagger}) Z_{[\ell-1]}\big) \\
		&\quad \, +\frac12 \tr_{1, \dots , \ell-1}\big( \tr_{\ell}( Z_{[\ell]}^{\dagger}\rho_{\ell}) \rho_{1} \otimes \cdots \otimes \rho_{\ell-1} Z_{[\ell-1]}\big) \\
		&= \frac{1}{2} \tr_{1, \dots , \ell-1}\big(\rho_{1} \otimes \cdots \otimes \rho_{\ell-1} Z_{[\ell-1]}^{\dagger} Z_{[\ell-1]}\big) \\
		&\quad \, + \frac 12 \tr_{1, \dots , \ell-1}\big( Z_{[\ell-1]}^{\dagger} \rho_{1} \otimes \cdots \otimes \rho_{\ell-1} Z_{[\ell-1]}\big) \\
		&= \big\| Z_{[\ell-1]} \big\|_{2, \rho_{1} \otimes \cdots \otimes \rho_{\ell-1}}^{2}.
	\end{align*}
By a similar argument we also have 
\begin{align*}
		\big\langle Z_{[\ell]},  Y_{\ell} \big\rangle_{\rho_{1} \otimes \dots \otimes \rho_{\ell}} = \frac{1}{\sqrt{n}} {\| Y_{\ell} \|}_{2, \rho_\ell}^{2}.
	\end{align*}
Finally, since $Z_{[\ell-1]}$ and $Y_{\ell}$ act on separate subsystems and $\tr\big( \rho^{\otimes n} \widetilde S_\sigma\big)=0$ we have
	\begin{align*}
		\big\langle Z_{[\ell-1]},  Y_{\ell} \big\rangle_{\rho_{1} \otimes \dots \otimes \rho_{\ell}}  & = 
		\big\langle Z_{[\ell-1]},  \mathbb  I_{[\ell-1]}  \big\rangle_{\rho_{1} \otimes \dots \otimes \rho_{\ell-1}} \cdot\langle \mathbb I_\ell , Y_\ell\rangle_{\rho_\ell}\\
		&=  \tr\big(\rho_1\otimes \cdots \otimes \rho_{\ell-1}Z^\dagger_{[\ell-1]}\big)\cdot \tr(\rho_\ell  Y_\ell) \\
		&= 0. 
	\end{align*}
Therefore,
	\begin{align*}
		\Big\| Z_{[\ell]} - &Z_{[\ell-1]} - \frac{1}{\sqrt{n}} Y_{\ell}\Big\|_{2, \rho_{1} \otimes \dots \otimes \rho_{\ell}}^{2} \\
		&= {\| Z_{[\ell]} \|}_{2, \rho_{1} \otimes \dots \otimes \rho_{\ell}}^{2} - {\| Z_{[\ell-1]} \|}_{2, \rho_{1} \otimes \cdots \otimes \rho_{\ell-1}}^{2} 
+\frac{1}{n} \| Y_{\ell} \|_{2, \rho_\ell}^{2} - \frac{2}{\sqrt{n}} {\| Y_{\ell} \|}_{2, \rho_\ell}^{2}\\
		&= {\| Z_{[\ell]} \|}_{2, \rho_{1} \otimes \dots \otimes \rho_{\ell}}^{2} - {\| Z_{[\ell-1]} \|}_{2, \rho_{1} \otimes \cdots \otimes \rho_{\ell-1}}^{2} +\frac{1}{n} \| Y_{1} \|_{2, \rho_1}^{2} - \frac{2}{\sqrt{n}} {\| Y_{1} \|}_{2, \rho_1}^{2}.
	\end{align*}
Using this in~\eqref{eq:apply-LowerRtildeLemma-S-sigma} and summing over $\ell$ imply that
\begin{align}\label{eq:sum-J-sigma-J-rho-bound}
	   \frac{\lambda_{\rho}(n-1)}{4 I(\rho)} \cdot \big\| Y_{1} + I(\sigma) \ac_{1}  -\langle S_{\sigma}^\dagger, S_\sigma\rangle_\sigma \ac^\dagger_{1} \big\|_{2, \rho_{1}}^{2} &\leq     \|\widetilde S_\sigma\|_{2, \rho^{\otimes n}}^2 + \| Y_{1} \|_{2, \rho_1}^{2}- 2\sqrt n {\| Y_{1} \|}_{2, \rho_1}^{2}\nonumber\\
	   & = \Big\| \widetilde{S}_{\sigma} -  \widehat Y \Big\|_{2, \rho^{\otimes n}}^2\nonumber\\
	   &\leq J(\sigma)-\frac{J(\sigma)^2}{ J(\rho)},
\end{align}
where the last inequality follows from~\eqref{eq:J-S-F-Y-bound}.
Our final step is to estimate the left hand side of the above inequality. 
We first note that
$$\langle Y_1, \ac_1^\dagger\rangle_{\rho_1} = \langle \widehat Y, \widehat \ac^\dagger\rangle_{\rho^{\otimes n}}  =\langle \widetilde S_\sigma, \widehat \ac^\dagger\rangle_{\rho^{\otimes n}} = \langle S_\sigma, \ac^\dagger\rangle_\sigma  = \tr\big([\ac, \sigma]^\dagger \ac^\dagger\big) = \tr\big(\sigma [\ac^\dagger, \ac^\dagger]\big)=0.$$
Moreover, since $\rho$ is in Williamson's form and the covariance matrix of $\rho$ is diagonal, we have $\langle  \ac_1, \ac_1^\dagger\rangle_{\rho_1}=0$. As a result,
$$\big\| Y_{1} + I(\sigma) \ac_{1}  -\langle S_{\sigma}^\dagger, S_\sigma\rangle_\sigma \ac^\dagger_{1} \big\|_{2, \rho_{1}}^{2}\geq \big\| Y_{1} + I(\sigma) \ac_{1}  \big\|_{2, \rho_{1}}^{2}.$$
On the other hand, we have
\begin{align*}
		\big\| Y_{1} + I(\sigma){\ac_{1}}  \big\|_{2, \rho_1}^{2}
		&= {\Big\| Y_{1} + \frac{1}{\mu}\ac_{1} + \Big(I(\sigma)- \frac{1}{\mu} \Big) {\ac_{1}} \Big\|}_{2, \rho_1}^{2} \\
		&= {\Big\| Y_{1} + \frac{1}{\mu}\ac_1 \Big\|}_{2, \rho_1}^{2} + J(\sigma)^2{\|{\ac_{1}}\|}_{2, \rho_1}^{2}  
		 + 2 J(\sigma) \text{Re}\Big\langle Y_{1} + \frac{1}{\mu}\ac_1, \ac_{1} \Big\rangle_{\rho_1} \\
		&= \Big\| Y_{1} + \frac{1}{\mu} {\ac_{1}} \Big\|_{2, \rho_1}^{2} + \mu J(\sigma)^2,
\end{align*}
where in the last line we use~\eqref{eq:S-a-inner-prod} and Proposition~\ref{Basictilde} to conclude that 
 \begin{align*}
 \Big\langle Y_{1} + \frac{1}{\mu}\ac_1, \ac_{1} \Big\rangle_{\rho_1} & = \Big\langle \widehat Y + \frac{1}{\mu} \widehat \ac, \widehat \ac \Big\rangle_{\rho^{\otimes n}}
 = \Big\langle \widetilde S_\sigma + \frac{1}{\mu} \widehat \ac, \widehat \ac \Big\rangle_{\rho^{\otimes n}}
= \Big\langle S_\sigma + \frac{1}{\mu}\ac, \ac \Big\rangle_{\sigma}
 = 0.
\end{align*}
Therefore, by~\eqref{RGUpperBound} we have
\begin{align*}
\big\| Y_{1} + I(\sigma){\ac_{1}}  \big\|_{2, \rho_1}^{2} &  = \Big\| Y_{1} + \frac{1}{\mu}{\ac_{1}} \Big\|_{2, \rho_1}^{2} + \mu J(\sigma)^2\\
& = \Big\|   \widehat Y + \frac{1}{\mu} \widehat \ac \Big\|_{2, \rho^{\otimes n}}^2 + \mu J(\sigma)^2\\
&\geq \frac{J(\sigma)^2}{J(\rho)}+ \mu J(\sigma)^2\\
& = \mu I(\rho) \frac{J(\sigma)^2}{J(\rho)}.
\end{align*}
Using this bound in~\eqref{eq:sum-J-sigma-J-rho-bound} yields 
$$\frac{\lambda_\rho \mu (n-1)}{4} \frac{J(\sigma)^2}{J(\rho)} \leq J(\sigma) - \frac{J(\sigma)^2}{J(\rho)},$$
that is equivalent to 
$$J(\sigma) \leq \frac{1}{1+ \frac{\lambda_\rho \mu }{4}(n-1)} J(\rho).$$
By Lemma~\ref{lem:Fisher-inf-bounds} this implies the desired inequality~\eqref{eq:Fisher-inf-bound-exact}.

In the multimode case, we need to apply Lemma~\ref{LowerRtildeLemma} on $X=S_{ \sigma, j}$ for every $1\leq j\leq m$. This gives us
\begin{align*}
	\frac{\lambda_{\rho} (\ell-1)}{ n I(\rho) } \bigg\| Y_{j,\ell} + \sum_{j'=1}^m\Big( {\langle S_{\sigma,j'}, S_{\sigma, j} \rangle}_{\sigma} \ac_{j',\ell}  & -\langle S_{\sigma, j}^\dagger, S_{\sigma, j'}\rangle_\sigma \ac^\dagger_{j',\ell}\Big) \bigg\|_{2, \rho_{\ell}}^{2} \\
	&\leq 	\Big\| Z_{j, [\ell]} - Z_{j, [\ell-1]} - \frac{1}{\sqrt{n}} Y_{j, \ell}\Big\|_{2, \rho_1 \otimes \dots \otimes \rho_{\ell}}^{2},
\end{align*}
which is slightly different from what we had in the single-mode case.  Nevertheless, since $\rho$ is in Williamson's form, we have $\langle \ac_j, \ac_{j'}\rangle_\sigma = \langle \ac_j, \ac_{j'}^\dagger\rangle_\sigma=\langle \ac^\dagger_j, \ac^\dagger_{j'}\rangle_\sigma= 0$ if $j\neq j'$. We also have
$$\langle Y_{j, \ell} , \ac_{j', \ell} \rangle_\sigma= \langle S_{\sigma, j}, \ac_{j', \ell}\rangle_\sigma = \tr([\ac_j, \sigma]^\dagger \ac_{j'})= \tr(\sigma [\ac_j^\dagger, \ac_{j'}])=0,$$
and similarly $\langle Y_{j, \ell} , \ac^\dagger_{j', \ell} \rangle_\sigma=0$.
Hence,
\begin{align*}
\bigg\| Y_{j,\ell} + \sum_{j'=1}^m\Big( {\langle S_{\sigma,j'}, S_{\sigma, j} \rangle}_{\sigma} \ac_{j',\ell} & -\langle S_{\sigma, j}^\dagger, S_{\sigma, j'}\rangle_\sigma \ac^\dagger_{j',\ell}\Big) \bigg\|_{2, \rho_{\ell}}^{2} \\
&\geq \big\| Y_{j,\ell} +  I_j(\sigma) \ac_{j,\ell}  - \langle S_{\sigma, j}^\dagger, S_{\sigma, j}\rangle_\sigma \ac_{j, \ell}^\dagger \big\|_{2, \rho_{\ell}}^{2}\\
&\geq \big\| Y_{j,\ell} +  I_j(\sigma) \ac_{j,\ell}  \big\|_{2, \rho_{\ell}}^{2},
\end{align*}
which is the same as what we get in the single-mode case. Continuing the proof, at the end we arrive at
$$J_j(\sigma) \leq \frac{1}{1+ \frac{\lambda_\rho \mu_j I_j(\rho)}{4 I(\rho)} (n-1) } J_j(\rho).$$
Once again using Lemma~\ref{lem:Fisher-inf-bounds} and summing over $j$ we obtain the desired bound.

\section{Final remarks} \label{secCo}

In this paper we proved three results regarding the convergence rate in quantum CLTs. Theorem~\ref{TraceMainTheoremM} involves quantum CLT in terms of trace distance. The convergence rate of $\mathcal O(n^{-1/2})$ established in this theorem seems tight in the case of $m=1$, yet it is desirable to relax the assumption in the multimode case and prove the same convergence rate under the assumption of finiteness of a lower-order moment, particularly third order. To this end, the idea of \emph{trunction} already applied in the classical case may be useful~\cite{BhRa}.  

In Theorem~\ref{MainTheorem} we proved the convergence rate of $\mathcal O(n^{-1})$ for the entropic quantum CLT. This convergence rate seems optimal, but under the strong assumption that the quantum state satisfies a Poincar\'e inequality. It is desirable to relax this assumption. To this end, it would be interesting to come up with quantum generalizations of ideas in~\cite{BCG} based on which the same convergence rate is proven in the classical case under much weaker assumptions.  

We also used the notion of SLD quantum Fisher information and proved Theorem~\ref{pseudoMainTheorem-v2} about the convergence rate of quantum CLT when the figure of merit is the SLD Fisher information distance. This theorem is again based on the assumption of Poincar\'e inequality. It is desirable to prove the same convergence rate under weaker assumptions.

The example of Subsection~\ref{subsec:opt-conv-rate} establishes the optimality of the convergence rate in Theorem~\ref{TraceMainTheoremM}  even under the stronger assumption that all the moments of the underlying quantum state are finite. Moreover, this example together with Pinsker's inequality confirm that the convergence rate of $\mathcal O(n^{-1})$ established in Theorem~\ref{MainTheorem} cannot be improved at least under the weaker assumption that all the moments are finite.

\paragraph{Acknowledgements.} The authors are thankful to Milad M. Goodarzi for several fruitful discussions and for his comments on an early version of the paper. This work is supported by the NRF grant NRF2021-QEP2-02-P05 and the Ministry of Education, Singapore, under the Research Centres of Excellence program.

\paragraph{Funding.} This work is supported by the NRF grant NRF2021-QEP2-02-P05 and the Ministry of
Education, Singapore, under the Research Centres of Excellence program.

{\small
\bibliographystyle{abbrvurl} %can also choose plain or abbrvurl
\bibliography{CLTBIB}
}

%************************************************
\appendix

\section{Properties of the Poincar\'e constant} \label{app:PoincareInequality}
In this appendix, we state some basic properties of the Poincar\'e constant $\lambda_\rho$.

\begin{proposition} \label{prop:PartialTracePoin-PoinTensor}
For any bipartite state $\rho_{12}$ we have $\lambda_{\rho_1}\geq \lambda_{\rho_{12}}$ where $\rho_1$ is the marginal state on the first subsystem.
\end{proposition}

\begin{proof}
Apply the Poincar\'e inequality for $\rho_{12}$ on operators $X_1$ that act only on the first subsystem.

\end{proof}

Any Gaussian unitary $U$ can be decomposed into \emph{single-mode squeezers,} \emph{phase shifters} and two-mode beam splitters~\cite{Serafini}. The latter two classes, namely phase shifters and two-mode beam splitters, generate the space of \emph{passive transformations} that have the property that their action commutes with the photon number operator $H_m=\sum_{j=1}^m \ac_j^\dagger \ac_j$. Passive transformations can be characterized in terms of their action on annihilation operators. Indeed, for any passive unitary $U$ there is an $m\times m$ unitary $(u_{ij})$  such that
\begin{align}\label{eq:U-passive}
U^\dagger \ac_j U = \sum_{j'=1}^m u_{jj'} \ac_{j'}.
\end{align}
We now prove a more interesting property of the Poincar\'e constant.

\begin{proposition} \label{PoinUnitary}
For any passive transformation $U$ we have $\lambda_{U \rho U^{\dagger}}=\lambda_\rho$. More generally, if $\lambda_\rho>0$, then for 
any Gaussian unitary $U$ we have
\[
	\lambda_{U \rho U^{\dagger}}>0.
\]
\end{proposition}

\begin{proof}
Let $U$ be passive, and suppose we want to verify the Poincar\'e inequality for $U \rho U^{\dagger}$ on operator $Y$. Let $X=U^\dagger Y U$.  Then, using~\eqref{eq:U-passive} and the fact that $(u_{ij})$ is unitary we have
\begin{align*}
\| \partial Y \|_{U\rho U^\dagger}^2  &= \sum_{j=1}^m  \Big(\big\| [\ac_j, Y]  \big\|_{U\rho  U^\dagger}^2  + \big\| [\ac^\dagger_j, Y]  \big\|_{U\rho  U^\dagger}^2\Big)\\
&= \sum_{j=1}^m  \Big(\big\| [U^\dagger \ac_j  U, X]  \big\|_{\rho  }^2  + \big\| [U^\dagger  \ac^\dagger_j U, X]  \big\|_{\rho }^2\Big)\\
&= \sum_{j=1}^m  \bigg(   \bigg\| \sum_{j'=1}^m u_{jj'}[ \ac_{j'}  , X]  \bigg\|_{\rho  }^2  + \bigg\| \sum_{j'=1}^m \bar u_{jj'}[ \ac_{j'}^\dagger, X]  \bigg\|_{\rho }^2\bigg)\\
&= \sum_{j=1}^m  \Big(   \big\| [ \ac_{j}  , X]  \big\|_{\rho  }^2  + \big\| [ \ac_{j}^\dagger, X]  \big\|_{\rho }^2\Big)\\
& = \|\partial X\|_{2, \rho}^2.
\end{align*}
It is also easily verified that $\|Y\|^2_{U\rho U^\dagger}  = \|X\|^2_{\rho}$ and $\tr(\rho X) = \tr(U \rho U^\dagger Y)$. Therefore, $\lambda_\rho \|X\|^2_{\rho}\leq \|\partial X\|_{2, \rho}^2$ implies $\lambda_\rho\|Y\|^2_{U\rho U^\dagger}\leq \| \partial Y \|_{U\rho U^\dagger}^2 $ and vice versa.  

To prove the second part we only need to handle the tensor product of single-mode squeezing unitaries which take the form $U = e^{\frac 12 \sum_{j=1}^m r_j(\ac_{j}^2 + (\ac_{j}^{\dagger})^2)}$. We claim that 
\[
	\lambda_{U \rho U^{\dagger}} \geq \frac{1}{c} \lambda_{\rho},
\]
where $c=2 \max_j \big(\cosh^2(r_j) + \sinh^2(r_j)\big)$.
To this end, we use  
	\[
		U \ac_{j} U^\dagger = \cosh(r_j) \ac_{j} + \sinh(r_j) \ac_{j}^{\dagger}.
	\]
Then, for $X=U^\dagger Y U$  we have
\begin{align*}
\| \partial X \|_{2, \rho}^2 
& = \| \partial X \|_{ U^\dagger (U \rho U^\dagger) U}^2\\
&  = \sum_{j=1}^m  \Big(   \big\| [ \ac_{j}  , X]  \big\|_{U^\dagger (U \rho U^\dagger) U  }^2  + \big\| [ \ac_{j}^\dagger, X]  \big\|_{U^\dagger (U \rho U^\dagger) U }^2\Big)\\
&  = \sum_{j=1}^m  \Big(   \big\| [ U\ac_{j} U^\dagger  , Y]  \big\|_{U \rho U^\dagger  }^2  + \big\| [ U\ac_{j}^\dagger U^\dagger, Y]  \big\|_{U \rho U^\dagger }^2\Big)\\
 & = \sum_{j=1}^m \Big( \big\| \cosh(r_j) [\ac_{j} , Y] + \sinh(r_j) [\ac_{j}^{\dagger}, Y]\big\|_{U \rho U^\dagger }^{2}  \\
 &\quad\qquad\qquad \quad+ \big\| \cosh(r_j) [\ac^\dagger_{j} , Y] + \sinh(r_j) [\ac_{j}, Y]\big\|_{U \rho U^\dagger }^{2}   \Big)   \\
&\leq  2\sum_{j=1}^m  \big(\cosh^2(r_j) + \sinh^2(r_j)\big) \Big( \big\|  [\ac_{j} , Y] \big\|_{U \rho U^\dagger }^{2}  + \big\|  [\ac^\dagger_{j} , Y] \big\|_{U \rho U^\dagger }^{2}   \Big)   \\
& \leq c \| \partial Y \|_{U\rho U^\dagger}^2.
\end{align*}
We also have $\|Y\|_{U\rho U^\dagger} = \|X\|_{2, \rho}$. Then, the desired inequality follows.

\end{proof}

The following corollary shows that the Poincar\'e constant does not decrease under convolution.

\begin{corollary} \label{PoinConvolution}
For any two quantum states $\rho_1,\rho_2$ and any $\eta \in [0,1]$ we have
\[
	\lambda_{\rho \boxplus_{\eta} \sigma} \geq \min\{\lambda_{\rho}, \lambda_{\sigma}\big\}.
\]
\end{corollary}

\begin{proof}
We note that $\rho \boxplus_{\eta} \sigma = \tr_2\big( U_\eta \rho\otimes \sigma U_\eta^\dagger \big)$, where $U_{\eta}$ is a beam splitter and is passive. Then, the bound on the Poincar\'e constant follows from Proposition~\ref{prop:PartialTracePoin-PoinTensor} and Proposition~\ref{PoinUnitary}.

\end{proof}

In the following proposition we aim to show that if a quantum state $\rho$ satisfies the quantum Poincar\'e inequality, then all of its moments are finite.

\begin{proposition}
Let $\rho$ be a quantum state satisfying $\lambda_\rho>0$. Then, all the moments of $\rho$ are finite. 
\end{proposition}

\begin{proof}
We need to show that for any integer $\kappa>0$ we have $\tr(  \rho H_m^\kappa )<+\infty$ where $H_m= \sum_{j=1}^m \ac_j^\dagger \ac_j$. This can be prove by a simple induction. Suppose that $\tr(  \rho H_m^{2\kappa-1} )<+\infty$. Then, applying  the Poincar\'e inequality for $X= H_m^{\kappa} - \tr(\rho H_m^{\kappa})$ we obtain $\lambda_\rho\|X\|_{2, \rho}^2\leq \|\partial  X\|_{2, \rho}^2$. We have 
$$\|X\|_{2, \rho}^2 = \|H_m^\kappa\|^2_\rho - \tr(\rho H_m^{\kappa})^2 = \tr(\rho H_m^{2\kappa}) - \tr(\rho H_m^{\kappa})^2,$$
and
$$\|\partial  X\|_{2, \rho}^2 = \|\partial  H_m^\kappa\|_{2, \rho}^2 =\sum_{j=1}^m \Big( \| [\ac_j, H_m^\kappa]  \|_{2, \rho}^2 + \| [\ac^\dagger_j, H_m^\kappa]  \|_{2, \rho}^2  \Big).$$
Now, the point is that expanding $H_m^\kappa$ and using the canonical commutation relations, we find that $\| [\ac_j, H_m^\kappa]  \|_{2, \rho}^2$ and $\| [\ac^\dagger_j, H_m^\kappa]  \|_{2, \rho}^2$ can be bounded in terms of $\|H_m^{\kappa-1}\|_{2, \rho}^2=\tr\big(\rho H_m^{2(\kappa-1)}\big)$. Using these in the Poincar\'e inequality gives $\tr(\rho H_m^{2\kappa})<+\infty$ as desired. 
\end{proof}

In the next proposition we consider the Poincar\'e inequality for Gaussian states.

\begin{proposition}
For any Gaussian state $\rho$ that is faithful we have $\lambda_\rho>0$.

\end{proposition}

\begin{proof}
By Proposition~\ref{PoinUnitary}, it suffices to consider Gaussian states that are in Williamson's form, i.e., Gaussian states of the form $\rho=\tau_1\otimes \cdots\otimes \tau_m$ where  $\tau_j = (1- e^{-\beta_j})e^{-\beta_j \ac_j^\dagger \ac_j}$ is a thermal state and $\nu_j = \frac{1+e^{-\beta_j}}{1-e^{-\beta_j}}\geq 1$, $j=1, \dots, m$ are the Williamson's eigenvalues of $\rho$. Since $\rho_G$ is faithful, $\nu_j>1$ for all $j$. In the following, we show that $\lambda_{\rho}>0$ in the case of $m=1$. The proof for arbitrary $m$ works similarly.

Recall that $D_{-z} D_w = e^{\frac 12(-z\bar w + \bar z w)}D_{w-z}$. Therefore, using~\eqref{eq:tau-chi-W} for $\tau=(1- e^{-\beta})e^{-\beta \ac^\dagger \ac}$ with $\nu = \frac{1+e^{-\beta}}{1-e^{-\beta}}> 1$ we have 
\begin{align*}
\tr\big(\tau D_{-z}D_w\big) & = ^{\frac 12(-z\bar w + \bar zw)} \tr\big( \tau D_{w-z} \big)\\
& = e^{\frac 12(-z\bar w + \bar zw)} \chi_\tau(w-z)\\
& = e^{\frac 12(-z\bar w + \bar zw)} e^{-\frac 12 \nu |w-z |^2}\\
& = \exp\Big(-\frac \nu2   (  |z|^2 + |w|^2 ) + \frac{\nu+1}{2}\bar z w + \frac{\nu-1}{2} z\bar w\Big ). 
\end{align*}
Let $r=|z|$, $z_0 = z/r$ and $s=|w|$, $w_0=w/s$. Also, define 
$$Q_{r, z_0} := e^{\frac \nu 2 r^2} D_{rz_0}.$$
Then, the above equation is equivalent to   
\begin{align*}
\tr\big( \tau Q_{r, z_0}^\dagger Q_{s, w_0} \big) &= \exp\bigg( \frac{rs}{2}\Big((\nu+1) \bar z_0 w_0 + (\nu-1) z_0 \bar w_0    \Big)     \bigg)\\
&= \exp\bigg( \frac{rs}{2}\Big((\nu+1) \bar z_0 w_0 + (\nu-1) \frac{1}{\bar z_0 w_0}    \Big)     \bigg).
\end{align*}
Starting with  $D_wD_{-z} = e^{\frac 12 (-\bar z w+ z\bar w)} D_{w-z}$, by a similar computation  we find that
\begin{align*}
\tr\big( \tau  Q_{s, w_0} Q_{r, z_0}^\dagger \big)
&= \exp\bigg( \frac{rs}{2}\Big((\nu+1) \frac{1}{\bar z_0 w_0} + (\nu-1) \bar z_0 w_0   \Big)     \bigg).
\end{align*}
Putting these together, we arrive at 
\begin{align*}
\big\langle Q_{r, z_0}, Q_{s, w_0}\big\rangle_\tau & = \frac{1}{2}\exp\bigg( \frac{rs}{2}\Big((\nu+1) \bar z_0 w_0 + (\nu-1) \frac{1}{\bar z_0 w_0}    \Big)     \bigg) \\
&\quad \, + \frac 12\exp\bigg( \frac{rs}{2}\Big((\nu+1) \frac{1}{\bar z_0 w_0} + (\nu-1) \bar z_0 w_0   \Big)     \bigg).
\end{align*}
Next, using the Taylor expansion of the exponential function we find that
\begin{align}
\big\langle  Q_{r, z_0}, &Q_{s, w_0}\big\rangle_\tau \nonumber
\\& = \frac 12 \sum_{k, \ell\geq 0} \frac{1}{2^{k+\ell} k! \ell!} (rs)^{k+\ell}\bigg( \Big(  \frac{\nu-1}{\bar z_0w_0}\Big)^k \big((\nu+1) \bar z_0w_0\big)^\ell   
+  \Big(  \frac{\nu+1}{\bar z_0w_0}\Big)^k \big((\nu-1) \bar z_0w_0\big)^\ell \bigg)\nonumber\\
& = \sum_{(m, n) \in \Sigma} c_{m, n} (rs)^m (\bar z_0 w_0)^{n},\label{eq:Q-inner-prod}
\end{align}
where 
$$\Sigma =\big\{(m, n):\, m\geq n\geq -m, m\cong n  (\text{mod } 2)  \big\},$$
and for every $(m, n)\in \Sigma$,
$$c_{m,n} = \frac{(\nu-1)^{\frac{m-n}{2}}(\nu+1)^{\frac{m+n}{2}} + (\nu-1)^{\frac{m+n}{2}}(\nu+1)^{\frac{m-n}{2}}  }{ 2^{m+1} \Big(\frac{m+n}{2}\Big)! \Big(\frac{m-n}{2}\Big)!  }> 0.$$

Recall that $D_{z} =e^{-\frac 12|z|^2} e^{z\ac^\dagger}e^{-\bar z\ac}$, which implies $Q_{r, z_0} = e^{\frac{1}{2}(\nu-1)r^2} e^{rz_0\ac^\dagger}e^{-r\frac{1}{z_0} \ac}$. Therefore, once again considering the Taylor expansion of the exponential function and computing $Q_{r, z_0}$, we obtain a sum over the set $\Sigma$. In fact, for any $(m,n)\in \Sigma$ there is an operator $B_{m,n}$ such that 
\begin{align}\label{eq:def-B-mn}
Q_{r, z_0} = \sum_{(m, n)\in \Sigma} \sqrt{c_{m,n}} r^m z_0^n B_{m,n}.
\end{align}
Then, we have 
\begin{align*}
\big\langle  Q_{r, z_0}, Q_{s, w_0}\big\rangle_\tau  & = \sum_{(m, n), (m', n')\in \Sigma} \sqrt{c_{m,n}c_{m',n'}} r^m s^{m'} \bar z_0^n w_0^{n'} \big\langle  B_{m,n}, B_{m',n'}\big\rangle_\tau.
\end{align*}
This holds for all real numbers $r, s$, and complex numbers $z_0, w_0$ with modulus $1$. Thus, comparing to~\eqref{eq:Q-inner-prod}, we find that 
$$\big\langle  B_{m,n}, B_{m',n'}\big\rangle_\tau = \delta_{m,m'}\delta_{n,n'}.$$
This means that operators $B_{m,n}$ for  $(m,n)\in \Sigma$ are orthonormal with respect to the inner product $\langle \cdot, \cdot\rangle_\tau$. On the other hand, since by Stone-von Neumann theorem
the span of displacement operators is dense, the operators $B_{m,n}$ span the whole space. We conclude that $\{B_{m,n}:\, (m,n)\in \Sigma\}$ is an orthonormal basis for the space $\mathbf L_2(\tau)$. 
We also note that by definition $B_{0,0} = \mathbb I$, so $B_{m,n}$ is orthogonal to the identity operator for any $ (m,n)\neq (0,0)$. 

Next, we use~\eqref{DisplacementOp} to conclude that $\big[\ac, Q_{r, z_0}\big] = rz_0 Q_{r, z_0}$.
Applying the expansion~\eqref{eq:def-B-mn} on both sides of this equation implies
\begin{align}\label{eq:ac-c-B}
\big[\ac, B_{m,n}\big] = 
\begin{cases}
\sqrt{\frac{c_{m-1, n-1}}{c_{m,n}}} B_{m-1, n-1} &\qquad (m-1, n-1)\in \Sigma,\\
 0 & \qquad \text{otherwise}.
\end{cases}
\end{align}
We also have $\big[\ac^\dagger, Q_{r, z_0}\big] = r\bar z_0 Q_{r, z_0}$ which gives
\begin{align}\label{eq:ac-dagger-c-B}
\big[\ac^\dagger, B_{m,n}\big] = 
\begin{cases}
\sqrt{\frac{c_{m-1, n+1}}{c_{m,n}}} B_{m-1, n+1} & \qquad (m-1, n+1)\in\Sigma,\\
0 & \qquad \text{otherwise}.
\end{cases}
\end{align}

Let $X\in \mathbf L_2(\tau)$ be an arbitrary operator satisfying $\langle \mathbb I, X\rangle_\tau = \langle B_{0,0}, X\rangle_\tau=0$. Then, there are coefficients $\alpha_{m,n}$ such that 
$$X = \sum_{(m,n)\in\Sigma \atop (m,n)\neq (0,0)} \alpha_{m,n} B_{m,n}.$$
Since $\big\{B_{m,n}:\, (m,n)\in \Sigma\big\}$ is orthonormal, we have
$$\|X\|_\tau^2= \sum_{(m,n)\in\Sigma \atop (m,n)\neq (0,0)} |\alpha_{m,n}|^2.$$
We use~\eqref{eq:ac-c-B} and~\eqref{eq:ac-dagger-c-B} to compute $\|\partial X\|_\tau^2 $. To this end, note that for $(0, 0)\neq (m, n)\in \Sigma$ we have $(m-1, n-1)\in \Sigma$ unless $n=-m$. Similarly, for $(0, 0)\neq (m, n)\in \Sigma$ we have $(m-1, n+1)\in \Sigma$ unless $n=m$. Therefore, letting $\Sigma_-:= \{(m,-m): m\geq 1\}$ and $\Sigma_+ := \{(m,m):\, m\geq 1\}$, we have 
\begin{align}\label{eq:norm-partial-X-nu}
\|\partial X\|_\tau^2 = \sum_{(m,n)\in\Sigma\setminus \Sigma_- \atop (m,n)\neq (0,0)} \frac{c_{m-1, n-1}}{c_{m,n}}  |\alpha_{m,n}|^2 + \sum_{(m,n)\in\Sigma\setminus \Sigma_+ \atop (m,n)\neq (0,0)}  \frac{c_{m-1, n+1}}{c_{m,n}} |\alpha_{m,n}|^2.
\end{align}
It is not hard to verify that $\frac{c_{m-1, n-1}}{c_{m,n}} \geq \frac{2}{\nu+1}$ for every $(0,0)\neq(m,n)\in\Sigma\setminus \Sigma_-$ and $\frac{c_{m-1, n+1}}{c_{m,n}}\geq \frac{2}{\nu+1}$ for every $(0,0)\neq(m,n)\in\Sigma\setminus \Sigma_+$. 
On the other hand, $\Sigma_-\cap \Sigma_+=\emptyset$.  Using these in~\eqref{eq:norm-partial-X-nu} we arrive at 
$$\|\partial X\|_\tau^2 \geq \frac{2}{\nu+1} \sum_{(m,n)\in\Sigma \atop (m,n)\neq (0,0)}  |\alpha_{m,n}|^2 =  \frac{2}{\nu+1}\|X\|_\tau^2.$$
Hence, $\tau$ satisfies the Poincar\'e inequality with constant $\lambda_\tau\geq \frac{2}{\nu+1}$.

\end{proof}

%*********************************************************

\section{Score operator for Gaussian states} \label{ModifiedGaussian}

In this appendix we compute the SLD score operator for a Gaussian state $\rho_G$ that is in Williamson's form. Such a state takes the form $\rho_G= \tau_1\otimes \cdots \otimes \tau_m$ where $\tau_j = (1-e^{-\beta_j}) e^{-\ac_j^\dagger \ac_j}$ is a thermal state. We also have 
$$\mu_j  = \| \ac_j\|_{2, \rho_G} = \| \ac_j\|_{2, \tau_j} = \frac{1+e^{-\beta_j}}{2(1-e^{-\beta_j})}.$$
We show that 
$$S_{\rho_G, j} = -\frac{1}{\mu_j} \ac_j.$$
Since $S_{\rho_G, j} =\pi_{\rho_G}^{\phi}([\ac_j, \rho_G])$, this equality is equivalent to $[\ac_j, \rho_G] = -\frac{1}{\mu_j}\pi_{\rho_G}^{\psi}(\ac_j)$.  Thus, we need to show  that
 $$ \ac_j\rho_G -  \rho_G \ac_j  =  -\frac{1}{2\mu_j} (\rho_G  \ac_{j} + \ac_j \rho_G),$$
 that is a simple consequence of $\ac_j \rho_G = e^{-\beta_j} \rho_G \ac_j$. 

%**********************************************************
\section{Details of the proof of Lemma~\ref{lem:LSI-modified-Fisher}} \label{ConstantFisherSobolev}

We claim that for constant
$$C=\frac{8e^{-3\beta/2}}{(1+e^{-\beta})^2},$$
we have $Cg(x, y)\leq h(x, y)$ for all $x, y\geq 0$ where
$$g(x, y) = \big(e^{\beta/4}\sqrt y - e^{-\beta/4}\sqrt x\big)^2,$$
and 
$$h(x, y) = \bigg( \sqrt{\frac{2}{x+y}}(y-x)  + \frac{1}{\mu} \sqrt{\frac{x+y}{2}} \bigg)^2,\qquad \quad \mu=  \frac{1+e^{-\beta}}{2(1-e^{-\beta})}.$$
With the change of variables $x'=e^{-\beta/4}\sqrt x, y'=e^{\beta/4}\sqrt y$ we have $g(x, y) = (y'-x' )^2$. We also compute
\begin{align*}
h(x, y) &= \bigg( \sqrt{\frac{2}{x+y}}(y-x)  + \frac{1}{\mu} \sqrt{\frac{x+y}{2}} \bigg)^2\\
& = \frac{2}{x+y} \bigg( (y-x)  + \frac{1}{\mu} \frac{x+y}{2} \bigg)^2\\
& = \frac{2}{x+y} \bigg( \frac{2}{1+e^{-\beta}} y  - \frac{2e^{-\beta}}{1+e^{-\beta}} x \bigg)^2\\
& = \frac{8e^{-\beta}}{(1+e^{-\beta})^2} \cdot \frac{1}{e^{\beta/2}{x'}^2+e^{-\beta/2}{y'}^2} ({y'}^2-{x'}^2)^2\\
& \geq  \frac{8e^{-\beta}}{(1+e^{-\beta})^2} \cdot \frac{1}{e^{\beta/2}(x'+y')^2} ({y'}^2-{x'}^2)^2\\
& =  \frac{8e^{-3\beta/2}}{(1+e^{-\beta})^2}  (y'-x')^2\\
& = C g(x, y).
\end{align*}
This proves the desired inequality.

\section{An extension of the Cauchy-Schwarz inequality}\label{CauchyShwarzT}

The following lemma is used in the proof of Lemma~\ref{LowerRtildeLemma}.

\begin{lemma}\label{lem:Cauchy-Scwarz}
Let $\rho_1, \rho_2$ be two quantum states let 
$T=\tr_1\big( \PiInner{\rho_{1}}{\psi}(A)  B  \big)$ where $A$ is an operator acting only on the first subsystem. Then we have $$\|T\|_{2, \rho_2} \leq \|A\|_{2, \rho_1} \cdot \|B\|_{2, \rho_1\otimes \rho_2}.$$
\end{lemma}

\begin{proof}
We assume with no loss of generality that $\|A\|_{2, \rho_1}=1$. Then, there is an orthonormal basis $\{A_k:\, k\geq 1\}$ for $\mathbf L_2(\rho_1)$ including $A_1=A$. Also, let $\{C_\ell:\, \ell \geq 1\}$ be an orthonormal basis for $\mathbf L_2(\rho_2)$. Now, expanding $B$ in the tensor product basis, we have
$$B = \sum_{k, \ell} a_{k, \ell} A_k\otimes  C_\ell,$$
for some $a_{k, \ell}\in \mathbb C$. Then, by the orthonormality of $\{A_k:\, k\geq 1\}, \{C_\ell:\, \ell\geq 1\}$ we have
\begin{align*}
    \|T\|_{2, \rho_2}^2 = \Big\| \sum_{\ell} a_{1, \ell} C_\ell \Big\|_{2,\rho_2}^2 
     = \sum_\ell |a_{1, \ell}|^2
     \leq   \sum_{k, \ell} |a_{k, \ell}|^2
     = \|B\|^2_{2, \rho_1\otimes \rho_2}.
\end{align*}
\end{proof}

\end{document}